\documentclass[10pt,twoside]{amundi_article}

\usepackage[table]{xcolor}
\usepackage{lmodern}
\usepackage{amssymb}
\usepackage{amsfonts}
\usepackage{amsmath}
\usepackage{amsbsy}
\usepackage{amsmath,amsfonts} 
\usepackage{amsthm}

\newtheorem*{Proposition}{Proposition}

\usepackage{fancybox}
\usepackage{fancyhdr}
\usepackage{graphicx}
\usepackage[hyphens]{url}
\usepackage{arydshln}
\usepackage{multirow}
\usepackage{pdflscape}
\usepackage{tikz}
\usepackage{comment}
\usepackage{nicefrac}
\usepackage{dsfont}
\usepackage{algorithmic}
\usepackage{algorithm}
\usepackage{lipsum}

\newlength{\figurewidth}
\newlength{\figureheight}
\DeclareMathAlphabet\mathbfcal{OMS}{cmsy}{b}{n}
\DeclareFontFamily{OT1}{pzc}{}
\DeclareFontShape{OT1}{pzc}{m}{it}{<-> s * [1.3] pzcmi7t}{}
\DeclareMathAlphabet{\mathpzc}{OT1}{pzc}{m}{it}

\usetikzlibrary{positioning}
\usetikzlibrary{patterns,arrows,decorations.pathreplacing}
\usetikzlibrary{backgrounds}

\usepackage[
        left=3.0cm,
        right=3.0cm,
]{geometry}

\newlength\FHoffset
\fancyheadoffset{\FHoffset}

\usepackage[authoryear]{natbib}

\usepackage{ragged2e}

\usepackage{subfigure}
\usepackage{booktabs}
\definecolor{amundi_blue}{RGB}{0,176,240}
\definecolor{amundi_dark_blue}{RGB}{0,28,75}
\definecolor{darkblue}{rgb}{0.0, 0.0, 0.55}
\usepackage{hyperref}
\hypersetup{
    colorlinks=true,
    citecolor=darkblue,
    linkcolor=darkblue,
    filecolor=darkblue,      
    urlcolor=darkblue
}

\usepackage{pdfpages}
\usepackage{bookmark}

\usepackage{pdfpages}
\usepackage{bookmark}

\usepackage{enumitem}
\usepackage{balance}

\usepackage{tabularx}

\usepackage{mathrsfs}

\begin{document}

\setcounter{page}{1}

\title{\textbf{\color{amundi_blue}Forecasting Probability of Default for Consumer Loan Management with Gaussian Mixture Models}%
}

\author{
\hspace{-0.3cm}
{\color{amundi_dark_blue} Hamidreza Arian\footnote{E-Mail: {\color{amundi_blue} \href{emailto:hamidreza.arian@utoronto.ca}{hamidreza.arian@utoronto.ca}}}} \\
\hspace{-0.3cm} Sharif University \\
\hspace{-0.3cm} of Technology
\and
{\color{amundi_dark_blue} Seyed Mohammad Sina Seyfi\footnote{E-Mail: {\color{amundi_blue} \href{emailto:sina.seyfi@aalto.fi}{sina.seyfi@aalto.fi}}}} \\
Sharif University \\
of Technology
\and
{\color{amundi_dark_blue} Azin Sharifi\footnote{E-Mail: {\color{amundi_blue} \href{emailto:azin.sharifi@utoronto.ca}{azin.sharifi@utoronto.ca}}}} \\
Sharif University \\
of Technology
}

\date{\color{amundi_dark_blue} October 2020}

\maketitle

\begin{abstract}
\noindent
Credit scoring is an essential tool used by global financial institutions and credit lenders for financial decision making. In this paper, we introduce a new method based on Gaussian Mixture Model (GMM) to forecast the probability of default for individual loan applicants. Clustering similar customers with each other, our model associates a probability of being healthy to each group. In addition, our GMM-based model probabilistically associates individual samples to clusters, and then estimates the probability of default for each individual based on how it relates to GMM clusters. We provide applications for risk managers and decision makers in banks and non-bank financial institutions to maximize profit and mitigate the expected loss by giving loans to those who have a probability of default below a decision threshold. Our model has a number of advantages. First, it gives a probabilistic view of credit standing for each individual applicant instead of a binary classification and therefore provides more information for financial decision makers. Second, the expected loss on the train set calculated by our GMM-based default probabilities is very close to the actual loss, and third, our approach is computationally efficient. 
\end{abstract}

\noindent \textbf{Keywords:} 
Credit Scoring, 
Machine Learning,  
Probability of Default, 
Gaussian Mixture Models, 
Financial Decision Making

\noindent \textbf{JEL classification:} C61, G11.

\section{Introduction}

In the face of the recent challenges, it has been a crucial issue for bankers and credit lenders to find models to predict borrower's bankruptcy with high precision (\cite{altman2014distressed}, \cite{veganzones2018investigation}). As a result, credit scoring systems have evolved as an essential part of bankruptcy prediction and credit risk management solutions. Financial institutions, especially banks, are widely using credit scoring models to allocate credit to good applicants, and to differentiate between good and bad borrowers. Applying credit scoring models will reduce credit transaction costs and the potential risk of a bad loan loss, and results in improvements in credit lending decisions. The main objective of credit risk analysis is to use personal and financial information to forecast default probability of individual customers, and to reduce the credit risk associated with issuing loan.  

Within the bankruptcy prediction studies that have been published, there is a wide range of classification methods used, including Linear Discriminant Analysis (LDA), Logistic Regression (LR) (\cite{dimitras1996survey}, \cite{ohlson1980financial}), and more recently, nonlinear techniques such as Support Vector Machines (SVMs) and Neural Networks (NNs) (\cite{lee1996hybrid}, \cite{huang2007credit}, and \cite{mancisidor2020deep}). Recently, advanced machine learning and data-driven approaches have gained popularity in offering state-of-the-art solutions for predicting bankruptcy. In addition, in machine learning applications in the corporate sector, the relationship between company features and credit decision-making is learned from the data rather than modeling the company's financial characteristics (\cite{atiya2001bankruptcy}). Most of the work in credit scoring has been focused on building models to provide a binary prediction on whether a customer will default on their loan or not. Such binary classification routines can be based on various types of supervised leaning algorithms (\cite{hand2001supervised}, \cite{zhang2010bayesian}, \cite{min2005bankruptcy}, \cite{salcedo2005genetic},  and \cite{wu2007real}).

\cite{hand1997statistical} review several statistical classification models for credit scoring. They emphasize that the credit scoring problem does not have a unique best solution. As a result of such surveys, numerous researchers have focused on providing hybrid and combined models for classification challenges present in building credit scoring solutions and some propose to combine domain knowledge with classification methods (\cite{sinha2008incorporating}). Some researchers implement a scorecard model using fewer classifiers while some other employ multiple classifiers and combine them with ensemble techniques (\cite{ala2016classifiers} and \cite{xiao2020cost}). To find an efficient ensemble classifier, \cite{hsieh2010data}, presents the class-wise classification. In order to investigate the performance of a combination of multiple classifiers with respect to a single classifier, \cite{tsai2008using} compare single and multiple classifiers for the bankruptcy prediction. They conclude multiple classifiers do not outperform the single ones in all cases. In contrary, \cite{wang2011comparative} show that multiple classifiers could outperform single classification algorithms. Their results suggest that three ensemble methods, namely Bagging, Boosting, and Stacking, can improve individual learners, like LR, Decision Trees, Artificial Neural Networks, and SVMs. In addition, \cite{zhang2010vertical} introduce vertical bagging decision trees model, which provides a combination of classifiers. SVMs are also capable of being a hybrid component with data mining techniques, where clustering and classification stages are performed through SVM (\cite{chen2012credit}). As another interesting research, Clustered Support Vector Machine (CSVM) for credit scoring has been reported as an efficient and flexible classifier method compared to other nonlinear SVM approaches (\cite{harris2015credit}). 

Despite the broad range of binary approaches used, there are techniques capable of estimating the probability of whether or not an individual would default on their loan (\cite{yeh2009comparisons}, \cite{capotorti2012credit}, \cite{gonen2012probabilistic}, \cite{kruppa2013consumer}, \cite{antunes2017probabilistic}, and \cite{feng2018dynamic}). Default probabilities provide more comprehensive consumer creditworthiness data than a binary or multi-category scoring. As most credit scoring articles are concerned with classification aspects of the subject, further study is needed to update scorecards with information on probability of default and forecasting profitability of the credit lending business itself (\cite{fang2019new}). Either case, for binary classification or calculating associated default probabilities, companies need to obtain individual applicants' information from commercial credit companies that gather customer information and use a mathematical model to infer a measure of propensity to default. By categorizing credit card users based on their usage into two groups of transactors and revolvers, \cite{so2014using} provide two separate scores for profit maximization. As an alternative approach to default probability estimation, \cite{serrano2016use} estimate the profitability by internal rate of return resulting in superior business profits. The purpose of this study is to introduce one of the most promising among recently developed statistical techniques – the Gaussian Mixture models – for default prediction and profit maximization.

In this paper, a new classification approach for predicting default probabilities is introduced. Our method is based on Gaussian Mixture model (GMM) clustering. GMMs are one of the most widely used unsupervised mixture models which are popular for their use in clustering applications (\cite{bishop2006pattern}). Using GMMs, we estimate a probability of default for each individual customer. All customers will be placed in different clusters with respect to their feature similarities. Upon clustering data by GMM, the ratio of good and bad customers on each cluster can suggest a probability of bankruptcy for members of that cluster. By using an appropriate decision threshold, according to risk tolerance of the loan issuer, all loan applicants can be classified in two groups.  
After the model is trained, given a new test sample, our model first calculates a probability that the new sample belongs to each cluster, and then uses the default probabilities of clusters to estimate how likely this new sample is to default. 
Training data is used for finding the optimal number of clusters, and the parameters of each cluster is estimated using Expectation-Maximization (EM) algorithm. According to these parameters, any sample belongs to a cluster with a given probability. As an interesting result from a risk managerial perspective, we show that our approach can predict Expected Loss (EL) of bad loans very well, even when the data is significantly imbalanced. Finally, we propose an approach for calculating a minimum probability decision threshold for a credit lender to use along with GMM-based default probabilities for individual customers for profit maximization.
  
In what follows, section \ref{sec:Methodology} provides the mathematical background of our approach for predicting default probabilities, where we discuss the Gaussian Mixture clustering, and how probabilities can be associated to these clusters. In section \ref{sec:app}, we show how GMM probabilities can be used to forecast EL for consumer loan portfolios and discuss the application of our model for profit maximization of credit institutions. Section \ref{sec:emp} provides information on the data we have used for our analysis and numerical results with regards to default probabilities and EL. Finally, section \ref{sec:conc} concludes the paper.

\section{Methodology} \label{sec:Methodology}

In this section, we introduce our probabilistic approach for estimating default probability of consumer loan holders. Our approach is based on Gaussian Mixture Model which is a popular unsupervised clustering technique but its applications in credit scoring has not been fully explored yet.

\subsection{Gaussian Mixture Models}

Gaussian Mixture models (GMM) are probability distribution functions defined as weighted summation of a finite set of normal distributions
\begin{equation}
p(\mathbf{x} | \boldsymbol{\theta})=\sum_{i=1}^{N_c} \omega_{i} \phi\left(\mathbf{x} | \boldsymbol{\mu_{i}}, \boldsymbol{\Sigma}_{i}\right) ,
\end{equation}
with the restriction of $\sum_{i=1}^{N_c} w_{i} = 1$. In our credit scoring application, $\mathbf{x}$ is the customer's features data, $N_c$ is assumed to be the number of clusters, $\boldsymbol{\omega} = \{\omega_i\}_{i=1}^{N_c}$ is the set of cluster weights, and $\boldsymbol{\mu} = \{\boldsymbol{\mu_i}\}_{i=1}^{N_c}$ and $\boldsymbol{\Sigma} = \{\boldsymbol{\Sigma_i}\}_{i=1}^{N_c}$ are the means and covariance matrices of clusters, respectively. $\phi$ also indicates a normal distribution density
\begin{equation}
{\phi\left(\mathbf{x} | \boldsymbol{\mu_{i}}, \boldsymbol{\Sigma_{i}}\right)= \frac{1}{(2 \pi)^{D / 2}\left|\boldsymbol{\Sigma_{i}}\right|^{1 / 2}}} { \exp \left\{-\frac{1}{2}\left(\mathbf{x}-\boldsymbol{\mu_{i}}\right)^{\prime} \boldsymbol{\Sigma_{i}}^{-1}\left(\mathbf{x}-\boldsymbol{\mu_{i}}\right)\right\}}
.
\end{equation}
Within the GMM framework, each cluster is uniquely defined with its parameters: ${\omega_i, \boldsymbol{\mu_i}, \boldsymbol{\Sigma_i}}$. 

\subsection{Probabilistic Clustering with Gaussian Mixture Model}

GMM is capable of dividing data points into clusters with respect to their statistical features. The clusters are represented by their means, covariances and their weights. Unlike k-means clustering, GMM provides a probabilistic association between samples and clusters, rather than a binary classification (\cite{bishop2006pattern}). For estimating this probabilistic association, a latent variable $\mathbf{z} = (z_1, z_2, \allowbreak \cdots, z_{N_c})$ is considered such that $z_i\in \{0, 1\}$ and $\sum_i{z_i}=1$, whereby $p(\mathbf{x}|z_k=1)=\phi(\mathbf{x}|\boldsymbol{\mu_k},\boldsymbol{\Sigma_k})$. The joint distribution can be calculated by $p(\mathbf{x},\mathbf{z}) = p(\mathbf{z})p(\mathbf{x}|\mathbf{z})$. GMM clustering is an unsupervised approach in which the clusters are not labeled. For credit scoring applications, our algorithm labels each GMM cluster according to its members and their associated probabilities. For calculating the probability of each point belonging to each cluster, we define the GMM-based probabilities as
\begin{equation} \label{eq:probability}
    p(\mathbf{x} \in C_j | \mathbf{\theta}) = \frac{\omega_j\phi(\mathbf{x}|\mathbf{\mu_j}, \mathbf{\Sigma_j})}{\sum_{i=1}^{N_c} \omega_{i} \phi\left(\mathbf{x} | \boldsymbol{\mu_{i}}, \boldsymbol{\Sigma}_{i}\right)},
\end{equation}
where $C_j$ represents the $j$th Gaussian cluster. 


\subsection{Parameter estimation}

Using training data set, we calibrate GMM parameter set $\boldsymbol{\theta} = \{\boldsymbol{\omega}, \boldsymbol{\mu}, \boldsymbol{\Sigma}\}$ to maximize the log-likelihood function $\sum_{i=1}^n \log p(\mathbf{x_i}|
\mathbf{\theta})$. Since in the case of GMM parameter estimation, maximum likelihood does not have a closed form solution, we use the Expectation-Maximization (EM) approach for estimating $\boldsymbol{\theta}$. For a faster convergence of EM algorithm, we use k-means initialization to find the first set of parameters. Then parameters will be updated in each step of the EM algorithm according to
\begin{equation}
\begin{aligned}
\boldsymbol{{\mu}_{k}^{\text { new }}} &=\frac{1}{n_k} \sum_{n=1}^{N} r\left(z_{n k}\right) \mathbf{{x}_{n}}, \quad {\omega}_{k}^{\text { new }} =\frac{\sum_{n=1}^{N} r\left(z_{n k}\right)}{N},\\ 
\mathbf{{\Sigma}_{k}^{\text { new }}} &=\frac{1}{n_k} \sum_{n=1}^{N} r\left(z_{n k}\right)\left(\mathbf{{x}_{n}}-\boldsymbol{{\mu}_{k}^{\text { new }}}\right)\left(\mathbf{{x}_{n}}-\boldsymbol{{\mu}_{k}^{\text { new }}}\right)^{\mathrm{T}},
\end{aligned}
\end{equation}
where $r(z_{n k})$ is the responsibility function
\begin{equation}\label{probfunc}
r\left(z_{n k}\right)=\frac{\omega_{k} \phi\left(\mathbf{{x}_{n}} | \boldsymbol{\mu_{k}}, \boldsymbol{\Sigma_{k}}\right)}{\sum_{j=1}^{N_c} \omega_{j} \phi\left(\mathbf{{x}_{n}} | \mathbf{\boldsymbol{\mu_j}}, \boldsymbol{\Sigma_j}\right)} . 
\end{equation}
When the log-likelihood function reaches the maximum number of iterations or when more iterations do not improve the corresponding log-likelihood function significantly, the algorithm terminates. 

\subsection{Probability of default for each cluster}

In our model, each individual credit applicant can be associated to a given cluster with the probability given in equation \ref{eq:probability}. Considering $n$ customers, with the target of $\mathbf{y} = \{y_1, \cdots, y_n\}$, with $y_i$ being either 0 (bad client) or 1 (good client), we define the pay-back probability of each cluster by
\begin{equation} \label{eq:probclus1}
    p(y_{C_j}=1 | \mathbf{\theta}) = \frac{\sum_{i=1}^{n} p(\mathbf{x}_i \in C_j | \mathbf{\theta})y_i}{\sum_{i=1}^{n} p(\mathbf{x}_i \in C_j | \mathbf{\theta})} ,
\end{equation}
where $ p(y_{C_j}=1 | \mathbf{\theta})$ is the probability of being healthy in the Gaussian cluster $C_j$ given $\mathbf{\theta}$. The probability of default for each cluster can be calculated from pay-back probability by $ p(y_{C_j} = 0 | \mathbf{\theta}) = 1- p(y_{C_j}=1 | \mathbf{\theta})$ as
\begin{equation} \label{eq:probclus0}
    p(y_{C_j}=0 | \mathbf{\theta}) = \frac{\sum_{i=1}^{n} p(\mathbf{x}_i \in C_j | \mathbf{\theta}) (1-y_i)}{\sum_{i=1}^{n} p(\mathbf{x}_i \in C_j | \mathbf{\theta})} .
\end{equation}

\subsection{Probability of default for individual applicants based on the clusters probabilities}

For calculating the probability of default for loan applicants in the test set, we use cluster-wise pay-back probabilities from equation \ref{eq:probclus1} across all clusters to estimate the pay-back probability, $ p(y_{\text{new}}=1|\mathbf{\theta})$, for a new applicant from the test set
\begin{equation} \label{eq:probtest1}
    p(y_{\text{new}}=1|\mathbf{\theta}) = \frac{\sum_{j=1}^{N_c}p(y_{C_j}=1 | \mathbf{\theta})p(\mathbf{x}_{\text{new}} \in C_j | \mathbf{\theta})}{\sum_{j=1}^{N_c}p(\mathbf{x}_{\text{new}} \in C_j | \mathbf{\theta})},
\end{equation}
where $\mathbf{x_{\text{new}}}$ is the feature vector of the customer in the test set and $y_\text{new}$ represents its corresponding target. Similarly, the default probability of $\mathbf{x_{\text{new}}}$ is calculated by $p(y_{\text{new}}=0|\mathbf{\theta}) = 1 - p(y_{\text{new}}=1|\mathbf{\theta})$ as
\begin{equation} \label{eq:probtest0}
    p(y_{\text{new}}=0|\mathbf{\theta}) = \frac{\sum_{j=1}^{N_c}p(y_{C_j}=0 | \mathbf{\theta})p(\mathbf{x}_{\text{new}} \in C_j | \mathbf{\theta})}{\sum_{j=1}^{N_c}p(\mathbf{x}_{\text{new}} \in C_j | \mathbf{\theta})},
\end{equation}
Equation \ref{eq:probtest0} uses the probabilistic information of all constructed clusters to compute the default probability of a new customer. 
Note that in both equations \eqref{eq:probtest1} and \eqref{eq:probtest0}, the denominator is equal to one and both equations can be simplified. However, we have kept the summation in the denominator for clarification purposes. The relationship between cluster-wise probabilities and new data features is displayed in Figure \ref{fig:flowchart2}.



\begin{figure}[ht!]
    \centering
    \includegraphics[width=0.9\linewidth]{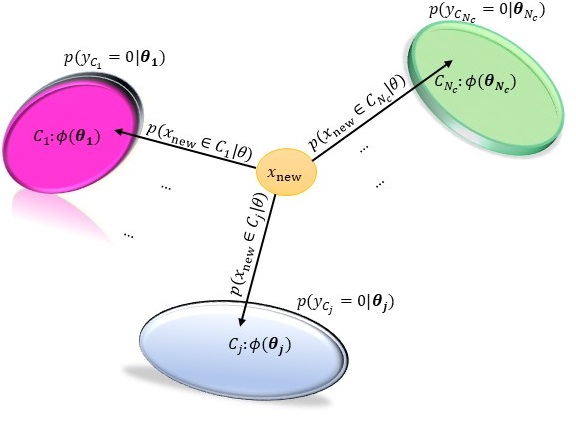}
    \caption{This graph shows the cluster-wise default probabilities, $p(y_{C_j}=0 | \mathbf{\theta})$ , as well as the probabilities that a new sample, like $\mathbf{x}_{\text{new}}$, belongs to each cluster, $p(\mathbf{x}_{\text{new}} \in C_j | \mathbf{\theta})$. The default probability calculated for $\mathbf{x}_{\text{new}}$ is the weighted average of the cluster-wise default probabilities weighted by the probability that $\mathbf{x}_{\text{new}}$ belongs to $C_j$, namely $\sum_{j=1}^{N_c}p(y_{C_j}=0 | \mathbf{\theta})p(\mathbf{x}_{\text{new}} \in C_j | \mathbf{\theta})$ (see equation \eqref{eq:probtest0} and note that $\sum_{j=1}^{N_c}p(\mathbf{x}_{\text{new}} \in C_j | \mathbf{\theta}) = 1$). This is worthwhile to mention that clusters are probabilistically overlapping as each sample belongs to all clusters with a given probability. The direction in which each cluster is spanned may also be different from others as each has its own covariance matrix. 
}
    \label{fig:flowchart2}
\end{figure}


\section{Applications to Risk Management} \label{sec:app}

In this section, we provide applications of our probability modeling technique for risk management of retail loans, including calculation of expected loss and a decision threshold for the probability of default. 

\subsection{Expected loss for loan issuer} \label{sebsec:app}

Expected loss (EL) is the summation of all possible losses that a loan issuer may bear due to default of borrowers. EL depends not only on probability of default, but also on the recovery rate (RR) and the Exposure at Default (EAD). Mathematically, the expected loss calculated by a predictive model is 
\begin{equation}\label{eq:expected_loss}
    \text{EL}_\text{model} = \sum_i \text{PD}_i \times (1 - \text{R}_i) \times \text{EAD}_i,
\end{equation}
and similarly, the actual expected loss from the data is 
\begin{equation}\label{eq:expected_loss}
    \text{EL}_\text{actual} = \sum_i {y_i \times (1 - \text{R}) \times \text{EAD}},
\end{equation}
where $R_i$ and $\text{EAD}_i$ are the RR and EAD for customer $i$. To assess the performance of our model for predicting loss due to bad loans, we define the relative EL error by
\begin{equation}
    \epsilon_\text{EL} = \frac{\text{EL}_\text{model}-\text{EL}_\text{actual}}{\text{EL}_\text{actual}},
\end{equation}
provided that the denominator is not zero. Here we will show that our terminology above for defining default probabilities based on GMM clusters is self-consistent and provides zero error for estimation of EL on the train set. 

\begin{Proposition}
On a consumer loan portfolio, suppose the probability of default is given by equation \eqref{eq:probtest0} and each customer has a constant loan amount and recovery rate. Then the GMM-based model provides perfect precision in predicting the expected loss on the train set, i.e., $\epsilon_{EL}=0$.
\end{Proposition}

\begin{proof}
Let us assume that constant capitals, $\text{EAD}_i = \text{EAD}$, and recovery rates, $\text{R}_i = \text{R}$, are associated to all customers. We will show that for the train set, EL calculated by equation \eqref{eq:expected_loss} perfectly matches the EL calculated by labels $\mathbf{y} = \{y_1, \cdots, y_n\}$ in the train set with $y_i$ being either 0 or 1. To see this mathematically, note that substituting $\text{PD}_i$ with $p(y_i=0|\theta)$ from equation \eqref{eq:probtest1} to equation \eqref{eq:expected_loss}, gives the model EL as
\begin{equation}
 \text{EL}_\text{model} = \sum_{i=1}^n {\sum_{j=1}^{N_c}p(y_{C_j}=0 | \mathbf{\theta})p(\mathbf{x}_i \in C_j | \mathbf{\theta})}\times (1 - \text{R}) \times \text{EAD}.
\end{equation}then by substitution from equation \eqref{eq:probclus0}, and by assuming probability $ p(y_{C_j}=0 | \mathbf{\theta})$ is associated to all members of cluster $C_j$, we have
\begin{equation}
\begin{aligned}
\text{EL}_\text{model} = & \sum_{i=1}^n {\sum_{j=1}^{N_c} \frac{\sum_{i=1}^{n} \Big( p(\mathbf{x}_i \in C_j | \mathbf{\theta}) (1-y_i) \Big)}{\sum_{i=1}^{n} p(\mathbf{x}_i \in C_j | \mathbf{\theta})} p(\mathbf{x}_i \in C_j | \mathbf{\theta})} \times (1 - \text{R}) \times \text{EAD}\\
= & \sum_{j=1}^{N_c} \Big( {\frac{\sum_{i=1}^{n} \Big( p(\mathbf{x}_i \in C_j | \mathbf{\theta}) (1-y_i) \Big)}{\sum_{i=1}^{n} p(\mathbf{x}_i \in C_j | \mathbf{\theta})}} \sum_{i=1}^n p(\mathbf{x}_i \in C_j | \mathbf{\theta}) \Big) \times (1 - \text{R}) \times \text{EAD}\\
= & \sum_{j=1}^{N_c} {{\sum_{i=1}^{n} \Big( p(\mathbf{x}_i \in C_j | \mathbf{\theta}) (1-y_i) \Big)}}\times (1 - \text{R}) \times \text{EAD} = \text{ EL}_\text{actual}
\end{aligned}
\end{equation}and therefore, $\epsilon_\text{EL}=0$ on the train data. 
\end{proof}

This is not a coincidence that when R and EAD are assumed constant, this error for the train set is $\epsilon_\text{EL} = 0$. The reason, as shown mathematically above, is due to the fact that GMM-based default probabilities are calculated for each cluster according to equation \eqref{eq:probclus0}, in such a way that it represents how many of the individual samples are labeled as bad borrowers. We will show a numerical example of the proposition above in section \ref{sec:expected_loss}. Since the assumption of equal EADs is not realistic, we also provide numerical examples when EADs follow a normal distribution with a given mean and standard deviation.

\subsection{Profit-maximizing based on credit scoring} \label{profit-maximizing}

As lenders aim to maximize their profit from providing credit to their customers, modern scorecard solutions are considering profit maximization into their credit scoring architecture (\cite{kozodoi2019multi}). Here in this section, we discuss how our GMM-based credit scoring model can assist issuers to maximize their profit. As maximizing the profit is equivalent to minimizing the loss, we analyze both the loss and the income for the lender. 

Suppose that the score of a consumer, with credit feature $\mathbf{x_i}$, is $s(\mathbf{x_i})$. Let $M_0$ be the capital lent to the borrower. We suppose $M$ is the amount the borrower is obligated to pay to the lender, and without loss of generality and for simplicity purposes, we assume that the interest rate's effect on the loan amount is embedded in $M$. If the customer defaults, the amount $MR_i$ is paid to the lender, where $R_i$ is the recovery rate for $i^{th}$ consumer. Let $p(y = 1|\mathbf{x_i})$ be the conditional probability of being \textit{good} given the feature set $\mathbf{x_i}$. Therefore, at time $T$, when the loan must be paid back, the customer $i$ pays the amount $\mathcal{I}(\mathbf{x_i})$ to the lender
\begin{equation}
\mathcal{I}(\mathbf{x_i}) = \left\{\begin{array}{ll}
0, & \text{if rejecting $\mathbf{x_i}$},\\
M, & \text{if accepting $\mathbf{x_i}$ and be good customer}, \\
MR_i, & \text{if accepting $\mathbf{x_i}$ and be bad customer}.
\end{array}\right.
\end{equation}

The mathematical expectation of the amount paid to the lender for giving loan to the customer $\mathbf{x_i}$ is
\begin{equation}
\begin{aligned}
    \mathbb{E}(\mathcal{I}(\mathbf{x_i})|\mathbf{x_i}) &= Mp(y = 1|\mathbf{x_i}) \\
    &+  MR_i(1-p(y = 1|\mathbf{x_i})).
\end{aligned}
\end{equation}
The net profit is therefore $\mathbb{E}(\mathcal{I}(\mathbf{x_i})|\mathbf{x_i}) - M_0$ and the total income can be calculated by
\begin{equation}
\begin{aligned}
      \sum_{i = 1}^n \mathbb{E}(\mathcal{I}(\mathbf{x_i})|\mathbf{x_i}) &=   \sum_{i = 1}^n [Mp(y = 1|\mathbf{x_i}) \\
    & +  MR_i(1-p(y = 1|\mathbf{x_i}))] \\
    &= M (\sum_{i=1}^n [p(y = 1|\mathbf{x_i})(1-R_i) + R_i] ).\\
\end{aligned}
\end{equation}
By assuming the recovery rates are constant, $R_i = R$ for $i = 1, \dots, n$, we can further simplify the total income as
\begin{equation}
      \sum_{i = 1}^n \mathbb{E}(\mathcal{I}(\mathbf{x_i})|\mathbf{x_i}) =  nMR + M(1-R) \sum_{i=1}^n p(y = 1|\mathbf{x_i}).
\end{equation}
Now let us define a minimum decision threshold $p_{\text{min}}$ for the lowest probability such that only customers with $p(y=1|x_i) \geq p_{\text{min}}$ are given loans. The expected income is
\begin{equation}
\begin{aligned}
    \sum_{i = 1}^n \mathbb{E}(\mathcal{I}(\mathbf{x_i})|\mathbf{x_i}) &= mMR + M(1-R) (\sum_{i=1}^m p(y = 1|\mathbf{x_i})) \\
    & \geq mMR + M(1-R) \sum_{i=1}^m p_{\text{min}} \\
    &= mM\left(R + p_{\text{min}} (1-R)\right),
\end{aligned}
\end{equation}
where $m$ is the number of customers with pay back probabilities higher than $p_{\text{min}}$. The above equation guarantees a lower bound for the lender's income by selecting $p_{\text{min}}$. 

The amount of loss for the customer $\mathbf{x_i}$, $\mathcal{L}(\mathbf{x_i})$, is
\begin{equation}
\mathcal{L}(\mathbf{x_i}) = \left\{\begin{array}{ll}
        0, & \text{if rejecting $\mathbf{x_i}$},\\
        0, & \text{if accepting $\mathbf{x_i}$ and be good customer},\\
         M(1-R_i), & \text{if accepting $\mathbf{x_i}$ and be bad customer},
        \end{array}\right.
\end{equation}
and the total expected loss for the lender from giving loan to $n$ customers, $\mathscr{L}(n)$, is therefore
\begin{equation}\label{eq:total_expected_loss}
 \begin{aligned}
 \mathscr{L}(n) =& \sum_{i=1}^n \mathbb{E}(\mathcal{L}(\mathbf{x_i})) = \sum_{i=1}^n M(1 - R_i) (1 - p(y = 1|\mathbf{x_i}))\\
 =& nM(1 - R)  - M(1 - R)\sum_{i=1}^n p(y = 1|\mathbf{x_i})),
 \end{aligned}
\end{equation} 
where in the last equation, we have assumed the recovery rates are constant across all clients, $R_i = R$. Equation \eqref{eq:total_expected_loss} provides a linear relationship between the lender's loss and the default probabilities. Obviously, the total loss is a function of $n$, the number of customers applying for a loan.



By choosing the decision threshold probability, $p_{\text{min}}$, we can consequently calculate an upper bound for the total expected loss as
\begin{equation} \label{eq:upper_bound}
\begin{aligned}
    \mathscr{L}(n) &= mM(1 - R)  - M(1 - R)\sum_{i=1}^m p(y = 1|\mathbf{x_i}))\\
    & \leq mM(1 - R) - mM(1-R) p_{\text{min}} \\
    &= mM(1 - R)(1 - p_{\text{min}}).
\end{aligned}
\end{equation}
Clearly, $m=m_{p_{\text{min}}}$ is a function of decision threshold $p_{\text{min}}$ and is related to the upper bound of the total expected loss, $\mathscr{L}(n)^*$, by the non-linear equation
\begin{equation} \label{eq:loss}
    \mathscr{L}(n)^* = mM(1 - R)(1 - p_{\text{min}}).
\end{equation}
This non-linear function can be solved numerically to give $p_\text{min}$ as a function of total expected loss
\begin{equation}
        p_{\text{min}} = \mathcal{F}(\frac{\mathscr{L}(n)^*}{M(1 - R)}),
\end{equation}
Therefore, $p_{\text{min}}$ can be uniquely determined by inputs $\mathscr{L}(n)^*$, $M$ and $R$. The non-linear function $\mathcal{F}$ depends on the relationship between the number of clients to whom loans are issued, $m$, and the decision threshold, $p_{\text{min}}$. We estimate this functional relationship by the train set data and then apply it to the test set. 

The average amount of loss for all applicants has the upper bound, $\mathscr{L}' = \frac{\mathscr{L}(n)^*}{n}$ given by
\begin{equation}
    \mathscr{L}' = \frac{m}{n}M(1 - R)(1 - p_{\text{min}}).
\end{equation}
This equation shows the amount of loss per customer which directly depends on the ratio of clients given a loan to all clients, $\frac{m}{n}$.

\section{Empirical Analysis}\label{sec:emp}

In this section, we provide our empirical analysis on a number of credit data sets to forecast the probability of default for each customer based on the methodology in chapter \ref{sec:Methodology}. We then use our model to classify clients, calculate the expected loss and the upper bounds for total loss based on estimated probabilities discussed in chapter \ref{sec:app}.

\subsection{Data}

We use a diversified group of five data sets from the UCI Repository of Machine Learning Database including German Credit Data, Australian Credit Approval and Japanese Credit Screening, Bank Marketing Data Set (Telemarketing) used by \cite{moro2014data}, default of credit card clients Data Set in Taiwan used by \cite{yeh2009comparisons}  (\cite{Dua:2019}), as well as a large data set from the U.S. Small Business Administration (SBA) used by \cite{li2018should} in some groups for our model.
 
The German Credit Data consists of 1000 customers. The Australian Credit Approval is composed of 307 good customers and 383 bad customers, and Japanese Credit Screening has 297 and 364 good and bad customers, respectively. The number of instances and features in Taiwanese data set is 30000 and 24 feature. For the telemarketing data-set, we have 45211 instances and 17 attributes. In this data set, the default is one of the variables which has been used as the target in our model. The SBA data set has 27 features and more than 300000 of instances are included in the our study.
 


\subsubsection{Data Preprocessing}

We transform categorical features to numerical features before inputting into the GMM. We use dummy variables for converting the categorical variables to numerical features. We also use a one-hot-encoder for German data set to convert the dummies to binary variables, increasing the number of features to 62. The German data is imbalanced, having 30\% unhealthy applicants against 70\% healthy ones. In order to balance this data set, we apply a Synthetic Minority Oversampling Technique (SMOTE). This method is also applied to other data sets except for the Australian and Japanese data sets, because they are fairly balanced. Note that SMOTE is only applied to the training sets. \cite{chawla2002smote} propose this method for handling the unbalanced data sets, suggesting to draw a line between samples which are same in label, and randomly generating new fake samples on that line. Therefore, if $\mathbf{x_1}$ and $\mathbf{x_2}$  are two such samples with the same label, then the new generated sample is
\begin{equation}
    \mathbf{x_{\text{generated}}} = (1-\alpha)\mathbf{x_1} + \alpha\mathbf{x_2}
\end{equation}
where $\alpha$ is a sample drawn from a uniform distribution $\mathcal{U}[0,1]$. By randomly selecting $\mathbf{x_1}$ and $\mathbf{x_2}$ and generating new samples enough times, the data set will be balanced. 

We split the data set to train and test sets with the ratio of $\frac{2}{3}$ and $\frac{1}{3}$. For data sets from SBA, because the number of instances are fairly high, we use the ratio of 50 \% for training and testing set. Using SMOTE, the number of healthy and unhealthy instances in the training set becomes equal. In Australian and Japanese data sets, there is no need to use SMOTE since they are fairly balanced. 

\subsection{Estimating the number of clusters}

To find the number of clusters, $N_c$, we use two criteria which are Akaike Information Criterion (AIC) and Bayesian Information Criterion (BIC)
\begin{equation}
\begin{aligned}
    AIC = & 2k - 2 \ln{\hat{L}},\\
    BIC = & \ln(n)k - 2 \ln{\hat{L}},
\end{aligned}
\end{equation}
where 
\begin{equation}
\hat{L} = p(\mathbf{x}|\mathbf{\theta}, N_c)
\end{equation}
is the maximized value of the likelihood function of the Gaussian Mixture model with $N_c$ components, $n$ is the number of the train set samples, $k$ is the number of model parameters, and $\mathbf{x}$ is a train sample. The AIC and BIC methods give the optimal number of GMM components when they reach their minimum level. Figure \ref{figAICBICall} shows the AIC and BIC values for different $N_c$s, proposing $N_c = 9$ for German, Japan and Australian sets, $N_c =80$ for Taiwanese data, $N_c = 55$ for telemarketing data, and $N_c =28$ for SBA case is used. We also used 90, 145, and 150 number of components for different sub-groups of SBA. 

The probability of default for individuals in both train and test sets is shown in figure \ref{fig:probability_default}. Our model, as figure \ref{fig:probability_default} shows, is capable of distinguishing good and bad applicants. As the train set of the German data set is imbalanced, it is clear from panels (a) and (b) of the figure that the model is not estimating the default probability of such samples correctly. However, as illustrated in panels (c), (d), (e) and (f), in Japanese and Australian data sets which are fairly balanced, the estimated probabilities of default for the train and test sets are along the same lines. 

\subsection{Measuring probability of default}

After creating $N_c$ clusters using the information criteria approach from last section, the results for the German, Australian and Japanese data sets are shown in table \ref{tab_clusters}. The GMM weights represent the relative importance of each cluster. In the forth row of each sub-table, we assign a probability $p(\mathbf{x_i}\in C_j|\boldsymbol{\theta})$ to a given sample with feature $\mathbf{x_i}$, indicating how likely $\mathbf{x_i}$ belongs to the cluster $C_j$. The ratio associated to each cluster showing the number of good to bad customers in that specific cluster is related to the cluster-based probability, computed by equation \ref{eq:probability}.

For the sake of simplicity, in order to classify customers, we set a decision threshold $D = \frac{1}{2}$. Each customer with $p(x|\theta) \geq D$ will be labeled as $1$, and otherwise as $0$. Therefore, our probabilistic model can also work as a binary classifier. We evaluate our model by computing the confusion matrix and accuracy measures of precision, recall, $F1$, and the ROC-AUC scores. The confusion matrix shows True Positives (TP), and True Negatives (TN), in the first row, and False Positives (FP) and False Negatives (FN) in the second row, respectively.

For a comparative analysis of our GMM-based approach, we use two other benchmarks; Support Vector Machine (SVM), and Logistic regression (LR). In our comparison, we use an SVM with radial basis kernel and a regularization parameter $C=1$. Model accuracy, as the ratio of the number of true predictions to the number of samples, are shown in table \ref{tab_accuracy2}. Looking at the table, it is clear that GMM, when compared to other two benchmarks, provides a reasonable level of accuracy on both train and test sets.

The SVM approach is not able to perform as well as GMM due to the fact that it can not separate unstructured data without normalization while GMM is a more flexible approach which is capable of exploring directional behavior of the data with its embedded covariance structure. In other words, in contrary to the GMM, the SVM cannot handle the non-normalized data. Since we used the same structure of data for all models, it affects the SVM's accuracy. The recall, precision, and the $F_1$ score for the above classifiers are shown in table \ref{tab_accuracy2}. 

\subsection{Expected loss}\label{sec:expected_loss}

The expected loss will be computed under two different conditions: first when the loan amounts are constant for all customers, and second when the loan amounts follow a normal distribution. For the first case, we suppose that the recovery rate for all customers is $0.5$, and the constant amount of loan is assumed to be \$$1000$. For loans with inconstant amounts, we assume that the retail portfolio of loans has individual loan amounts which follow a normal distribution with a certain mean and standard deviation. More precisely, we assume that each customer's loan amount is independently drawn form a normal distribution with a mean of \$$1000$  and a standard deviation of \$$100$. In each case, we calculate the expected loss using our GMM-based probabilistic model, and compare it with the actual loss for both train and test sets. As discussed in section \ref{sebsec:app}, in the train set when the loan amount is constant, the actual loss and the expected loss are the same. Due to some reasonable errors in our GMM-based probability calculation, the actual and expected loss amounts differ on the test set. The frequency of customers' actual loss and predicted loss for train and test sets, with fixed and random amounts of loans are shown in figure \ref{fig:fixed&random_loan}. In each panel in figure \ref{fig:fixed&random_loan}, the upper left quarter shows the predicted loss with $M = 1000$, and the lower left quarter demonstrates the bar chart of actual losses with the same amount of $M$. In the upper right quarter, the bar chart of predicted loss when $M$ is random is presented. The bar chart of actual loss for a random $M$ is represented in the lower left quarter. It is clear that the predicted loss for those defaulting, is less than the actual loss, while the predicted loss for customers not defaulting is more than $0$. GMM can properly forecast default probabilities and models the distribution of the loss. When compared with the real outcomes of 0 or 1 in the train set, our GMM-based model has a zero error on the train set, and provides acceptable estimations of the expected loss on the test set.

For calculating the probability with SVM, we use the method proposed by \cite{platt1999probabilistic}, converting the SVM's scores to probabilities using the Sigmoid function. We use the same loan amounts and recovery rates as for the GMM-based model to find the expected loss for SVM and LR. The actual loss and expected loss for each model and the differences are shown in table \ref{tab_prob_loss}. 
When the amount of loan is constant, the error of predicted loss for GMM is significantly lower than that predicted by SVM and LR, showing that the GMM is working well with biased data in comparison to SVM and LR. When the amount of loan is random, the GMM's error is less than SVM and LR's in the train set.  In the German test set, the GMM's prediction error is significantly less than other models. For Japanese, Australian, Taiwanese test sets the GMM's error is also less than SVM and LR. For the SBA datasets, however, the other models are working better than GMM. 


\subsection{Managing the loss}

In this section, we provide numerical results of our decision threshold selection technique and show how it can be  used by a credit lender to manage the amount of the loss due to bad loans. The main objective here is to find a decision threshold $p_{\text{min}}$ for managerial decision making to give loans to customer with feature $\mathbf{x_i}$ when $p(y=1|\mathbf{\theta}) \geq p_{\text{min}}$. As shown in equation \ref{eq:loss}, the amount of $\mathscr{L}(n)^*$ depends on the number of customers, loan amounts, recovery rate and the decision threshold. In figure \ref{fig:mn}, we investigate the relationship between $\frac{m}{n}(1-p_{\text{min}})$ and $p_{\text{min}}$, where $m$ is the number of customers having $p(y=1|\mathbf{x_i}) \geq p_{\text{min}}$. We also calculate the average loss per customer as $\mathscr{L}(n)' = \frac{\mathscr{L}(n)^*}{n}$.

By assuming $M(1-R)$ is constant, $\mathscr{L}(n)'$ is linearly proportional to $\frac{m}{n}(1-p_{\text{min}})$. Given this value for the train set, we can apply it for the test set and calculate the loss. Figure \ref{fig:mn} provides a comparison on train and test sets. As shown, $\frac{m}{n}(1-p_{\text{min}})$ for train and test sets in Australia and Japan data sets are very close to each other, and provide an accurate prediction of loss for test set. In the German data set, however, since the data is imbalanced, we observe a lag between two lines of train and test sets. Giving loans based on $p_{\text{min}}$, the actual average loss $\mathscr{L}(n)'$ for both train and test sets are shown in figure \ref{fig:loss_fix}. The accuracy of the predicted losses depends on the relationships estimated in figure \ref{fig:mn}.

We calculate an upper bound for the loss according to equation \ref{eq:upper_bound}, for both train and test sets. The results are shown in figure \ref{fig:upper_bound}. As this figure shows, the actual loss is below the upper bound we estimated in the train set (panels a, b and c). In the test sets (panels d, e and f),  the actual loss is under the upper bound, for decision thresholds below $p_{\text{min}} = 0.8$ for Australian and Japanese data sets, and $p_{\text{min}} = 0.7$ for German data set.
\section{Conclusion and Future Research}\label{sec:conc}

We proposed a probabilistic unsupervised learning approach for predicting default probability of consumer credit borrowers using GMM. Contrary to many other supervised learning approaches, we have used a flexible and efficient unsupervised learning technique. We used Gaussian Mixture clustering to group similar samples and then associated a cluster-wise probability to each group. From the default probability defined on the cluster level, we could derive default probabilities for individual samples. 

Our GMM-based probabilistic model has interesting applications in risk management of retail loan portfolios. Thanks to the probabilistic nature of our approach, we could give the user the flexibility to choose a decision threshold as a lower probability bound for decision makers to use for credit allocations to maximize the profit of their credit business. In addition, we showed that using our GMM-based default probabilities for calculating the expected loss due to bad loans, we can forecast the loss amount very well in both train and test sets. 

As there has been little work on the applications of Gaussian Mixture models in the credit scoring literature, for future research, we think that like other previously published supervised techniques such as neural networks, support vector machines and decision trees, it might be interesting to see the benefits of using hybrid GMM-based models and how their probabilistic nature benefits the literature. Another avenue of research could pursue the applications of GMM-based techniques in other classification problems in business such as customer churn prediction. 

\clearpage

\setcitestyle{numbers} 
\bibliographystyle{plainnat}
\bibliography{refs}

\clearpage

\appendix
\section*{Appendix}
\section{Numberical Results on German, Japenese and Australian Data Sets}




\begin{figure}[h!]
\centering
\subfigure[Germany]{{\includegraphics[width=3.5cm]{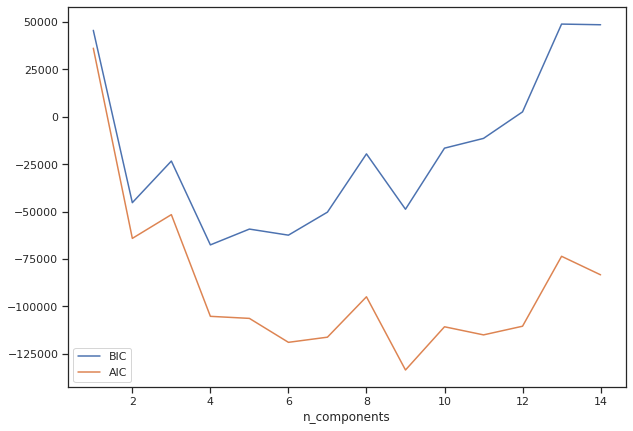}}}
\quad
\subfigure[Japan]{{\includegraphics[width=3.5cm]{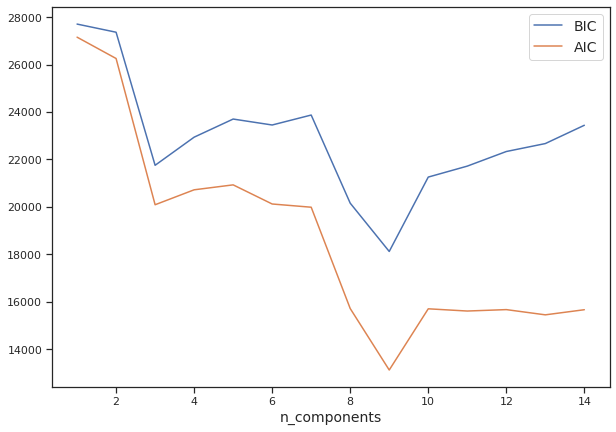} }}
\quad
\subfigure[Australia]{{\includegraphics[width=3.5cm]{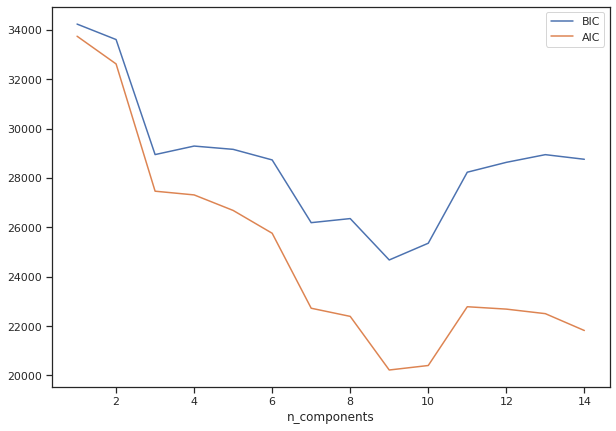} }} \\
\subfigure[Telemarketing]{{\includegraphics[width=3.5cm]{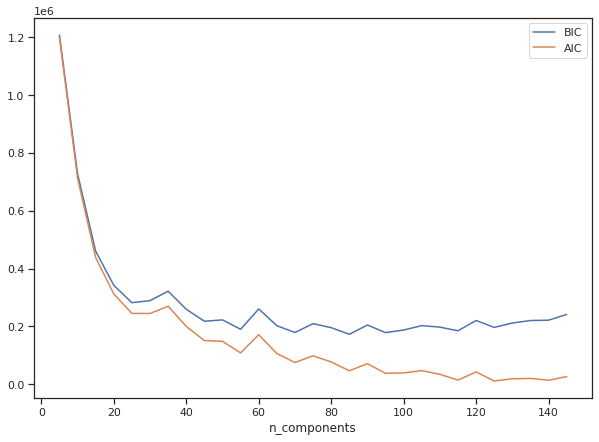}}}
\quad
\subfigure[Taiwanese]{{\includegraphics[width=3.5cm]{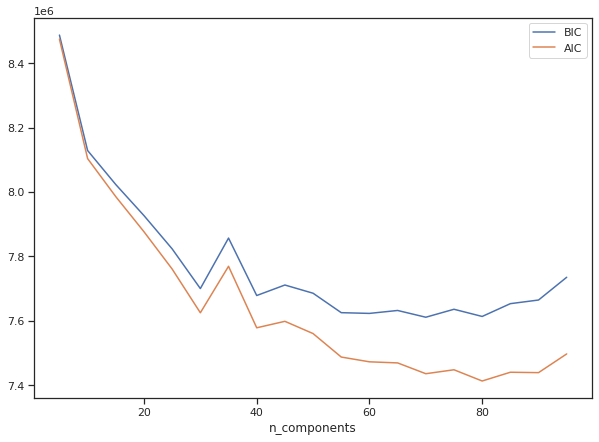} }}
\quad
\subfigure[SBAcase]{{\includegraphics[width=3.5cm]{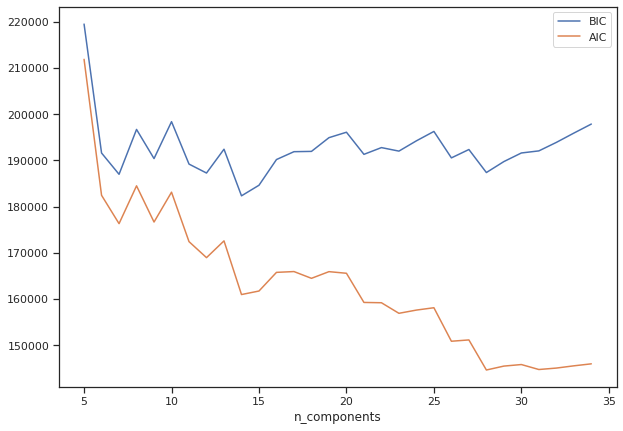} }}\\
\subfigure[SBA 1]{{\includegraphics[width=3.5cm]{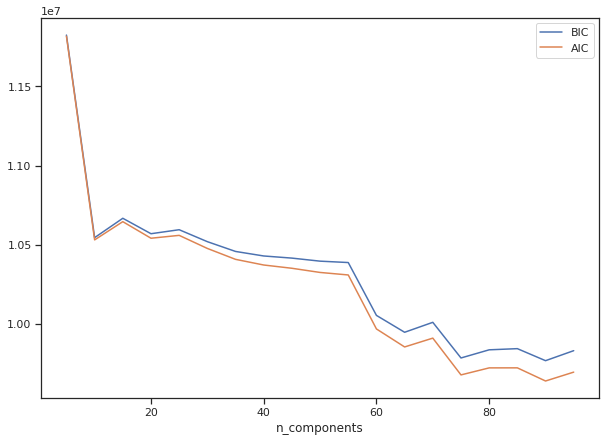}}}
\quad
\subfigure[SBA 2]{{\includegraphics[width=3.5cm]{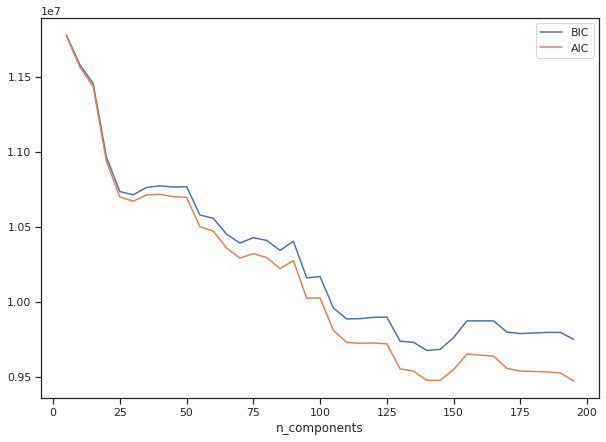} }}
\quad
\subfigure[SBA3]{{\includegraphics[width=3.5cm]{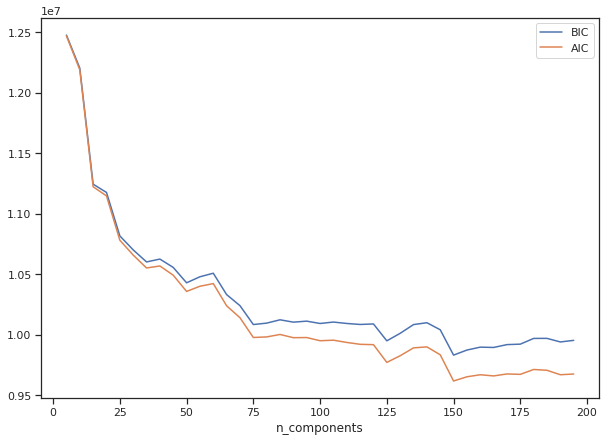} }}
\caption{AIC and BIC measures for optimal number of components (horizontal axis: the number of Gaussian components, vertical axis: the information criteria).}
\label{figAICBICall}%
\end{figure}


    

\begin{table}
\centering
\caption{Clusters' features in the German, Japanese, and Australian Data sets}
\begin{tabular}{llccccccccc}
\toprule
 \multicolumn{2}{l}{\multirow{2}{*}{ \textbf{German Data Set}}} & \multicolumn{9}{c}{Clusters} \\
 \cmidrule{3-11}
          &  & \textbf{1} & \textbf{2} & \textbf{3} & \textbf{4} & \textbf{5} & \textbf{6} & \textbf{7} & \textbf{8} & \textbf{9} \\
    \midrule
    \multirow{5}[2]{*}{\textbf{Train}} & \textbf{Good Customers} & 31    & 20    & 2     & 2     & 76    & 9     & 20    & 31    & 267 \\
          & \textbf{Bad Customers} & 20    & 26    & 202   & 9     & 26    & 18    & 16    & 23    & 118 \\
          & \textbf{Ratio} & 1.55  & 0.77  & 0.01  & 0.22  & 2.92  & 0.50  & 1.25  & 1.35  & 2.26 \\
          & \textbf{Probability} & 0.61  & 0.43  & 0.01  & 0.18  & 0.75  & 0.33  & 0.56  & 0.57  & 0.69 \\
          & \textbf{Weight} & 0.06  & 0.05  & 0.22  & 0.01  & 0.11  & 0.03  & 0.04  & 0.06  & 0.42 \\
    \cmidrule{2-11}
    \multirow{3}[1]{*}{\textbf{Test}} & \textbf{Good Customers} & 2     & 0     & 1     & 0     & 72    & 0     & 0     & 6     & 161 \\
          & \textbf{Bad Customers} & 1     & 0     & 0     & 0     & 39    & 0     & 0     & 0     & 52 \\
          & \textbf{Ratio} & 2.00  & --    & --    & --    & 1.85  & --    & --    & --    & 3.10 \\
\\
    \midrule
 \multicolumn{2}{l}{\multirow{2}{*}{\textbf{Japanese Data Set}}} & \multicolumn{9}{c}{Clusters} \\
 \cmidrule{3-11}
          &  & \textbf{1} & \textbf{2} & \textbf{3} & \textbf{4} & \textbf{5} & \textbf{6} & \textbf{7} & \textbf{8} & \textbf{9} \\
    \midrule
    \multirow{5}[2]{*}{\textbf{Train}} & \textbf{Good Customers} & 11    & 1     & 147   & 3     & 1     & 18    & 7     & 2     & 16 \\
          & \textbf{Bad Customers} & 40    & 0     & 24    & 0     & 0     & 42    & 0     & 0     & 128 \\
          & \textbf{Ratio} & 0.28  & --    & 6.13  & --    & --    & 0.43  & --    & --    & 0.13 \\
          & \textbf{Probability} & 0.21  & 1.00  & 0.86  & 1.00  & 1.00  & 0.30  & 1.00  & 1.00  & 0.11 \\
          & \textbf{Weight} & 0.12  & 0.00  & 0.39  & 0.01  & 0.00  & 0.14  & 0.02  & 0.00  & 0.33 \\
    \cmidrule{2-11}
    \multirow{3}[1]{*}{\textbf{Test}} & \textbf{Good Customers} & 3     & 0     & 72    & 0     & 0     & 7     & 0     & 0     & 9 \\
          & \textbf{Bad Customers} & 24    & 0     & 16    & 0     & 0     & 18    & 0     & 0     & 72 \\
          & \textbf{Ratio} & 0.13  & --    & 4.50  & --    & --    & 0.39  & --    & --    & 0.13 \\
    \\
    \midrule
 \multicolumn{2}{l}{\multirow{2}{*}{\textbf{Australian Data Set}}} & \multicolumn{9}{c}{Clusters} \\
 \cmidrule{3-11}
          &  & \textbf{1} & \textbf{2} & \textbf{3} & \textbf{4} & \textbf{5} & \textbf{6} & \textbf{7} & \textbf{8} & \textbf{9} \\
    \midrule
    \multirow{5}[2]{*}{\textbf{Train}} & \textbf{Good Customers} & 40    & 2     & 2     & 4     & 10    & 19    & 2     & 7     & 120 \\
          & \textbf{Bad Customers} & 118   & 0     & 0     & 0     & 4     & 114   & 0     & 7     & 11 \\
          & \textbf{Ratio} & 0.34  & --    & --    & --    & 2.50  & 0.17  & --    & 1.00  & 10.91 \\
          & \textbf{Probability} & 0.25  & 1.00  & 1.00  & 1.00  & 0.71  & 0.14  & 1.00  & 0.50  & 0.92 \\
          & \textbf{Weight} & \multicolumn{1}{r}{\textcolor[rgb]{ .129,  .129,  .129}{0.34}} & 0.00  & 0.00  & 0.01  & 0.03  & 0.29  & 0.00  & 0.03  & 0.28 \\
    \cmidrule{2-11}
    \multirow{3}[1]{*}{\textbf{Test}} & \textbf{Good Customers} & 20    & 0     & 0     & 0     & 2     & 13    & 0     & 0     & 66 \\
          & \textbf{Bad Customers} & 54    & 0     & 0     & 0     & 1     & 62    & 0     & 0     & 12 \\
          & \textbf{Ratio} & 0.37  & --    & --    & --    & 2.00  & 0.21  & --    & --    & 5.50 \\    
    \bottomrule
\end{tabular}%
\label{tab_clusters}%
\end{table}%



\begin{figure}
\centering
\subfigure[Telemarketing Train set]{{\includegraphics[width=5.61cm]{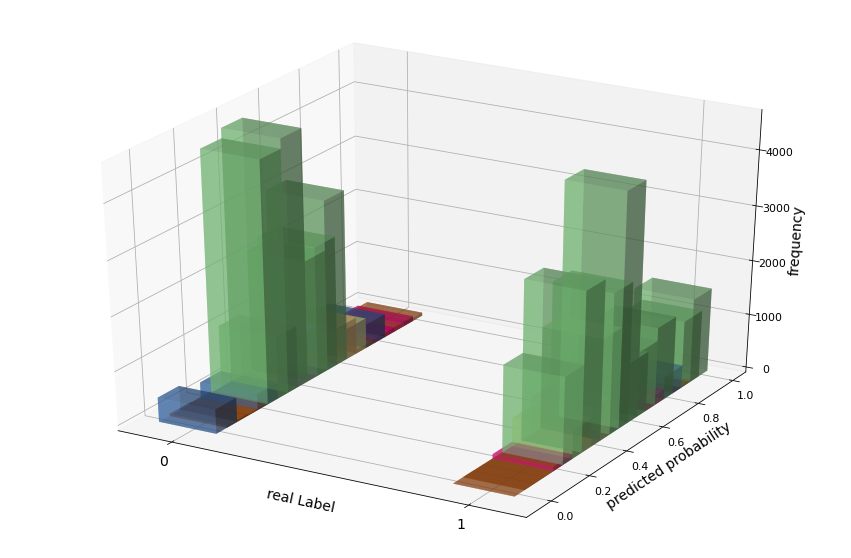}}}
\subfigure[Telemarketing Test Set]{{\includegraphics[width=5.61cm]{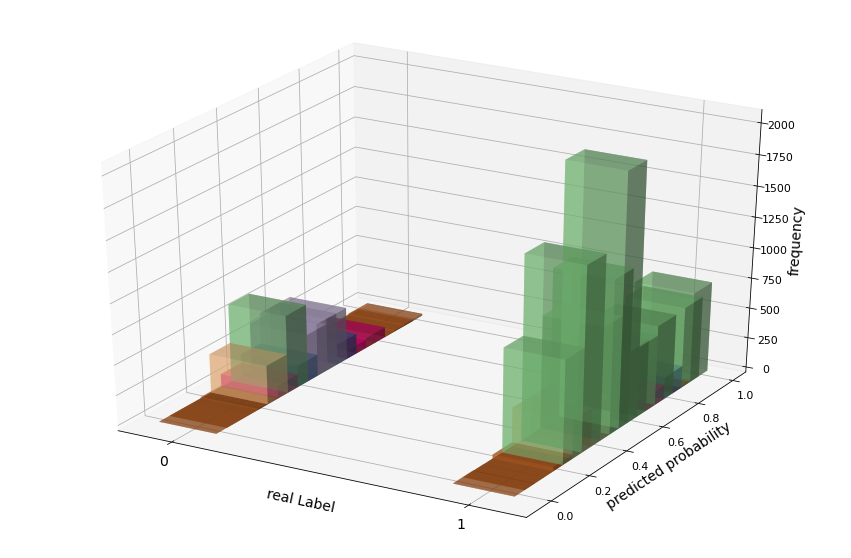} }}\\
\subfigure[Taiwanese Train set]{{\includegraphics[width=5.61cm]{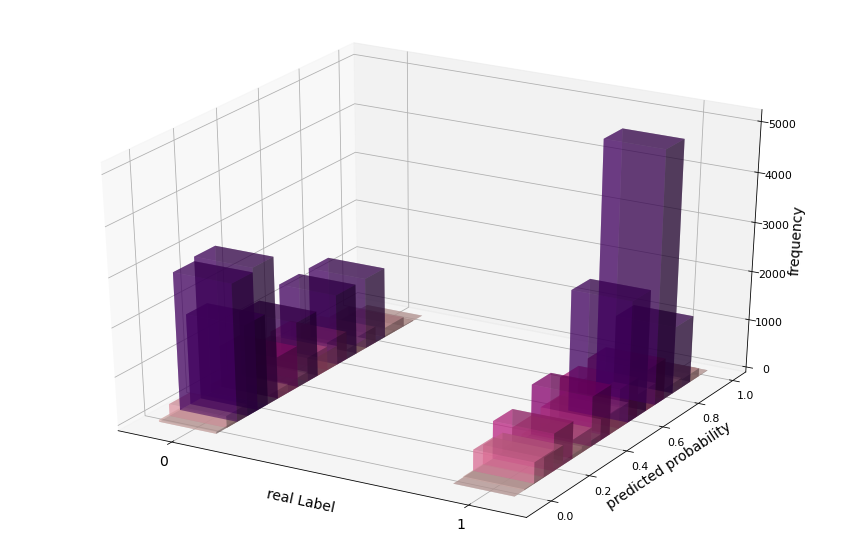}}}
\subfigure[Taiwanese Test Set]{{\includegraphics[width=5.61cm]{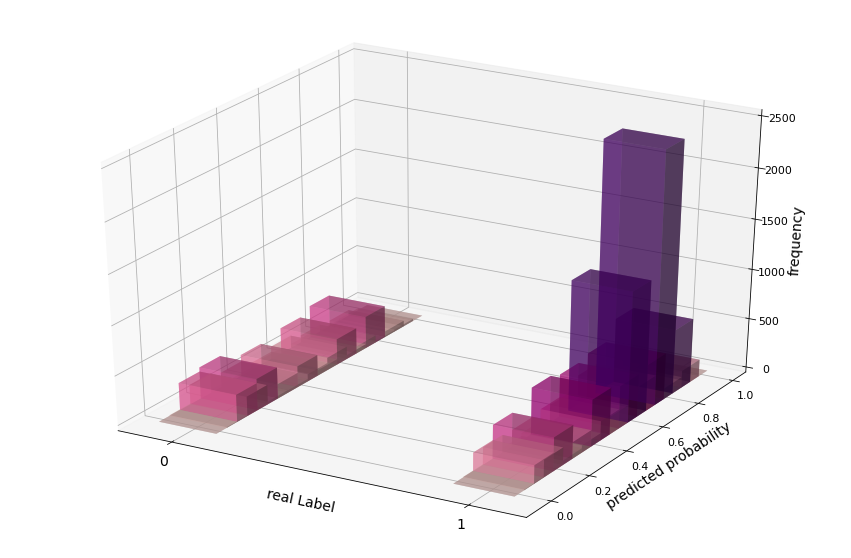} }}\\
\subfigure[SBA-case Train set]{{\includegraphics[width=5.61cm]{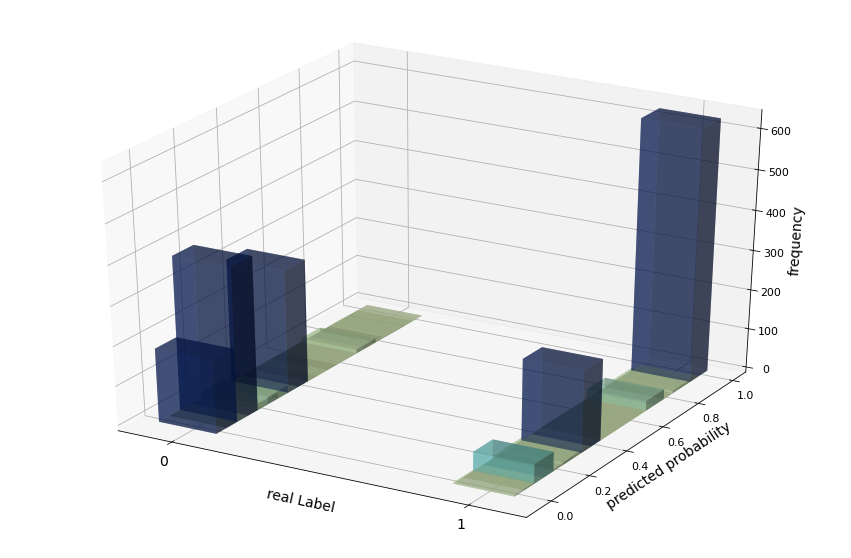}}}
\subfigure[SBA-case Test Set]{{\includegraphics[width=5.61cm]{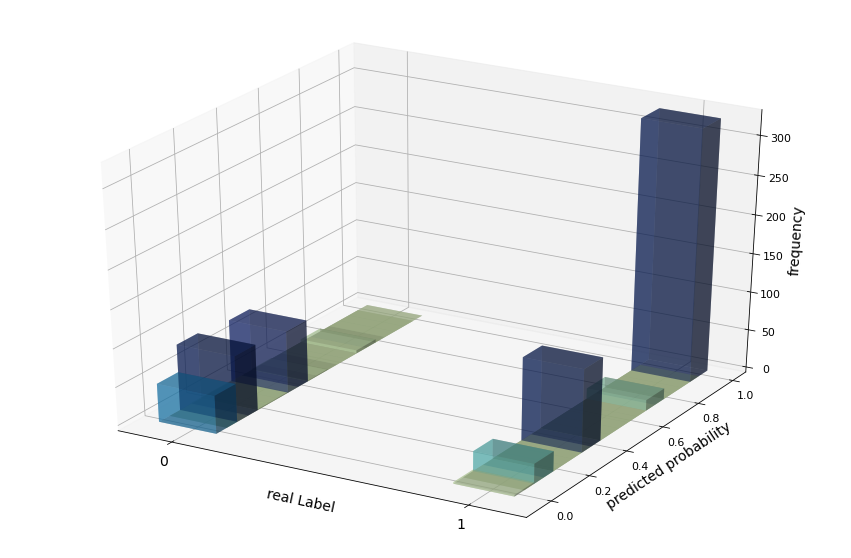} }}%
\caption{The probability of default for individuals Telemarketing, Taiwanese and SBA-case data sets.}
\label{fig:probability_default}
\end{figure}

\begin{figure}
\centering
\subfigure[Germany Train set]{{\includegraphics[width=5.61cm]{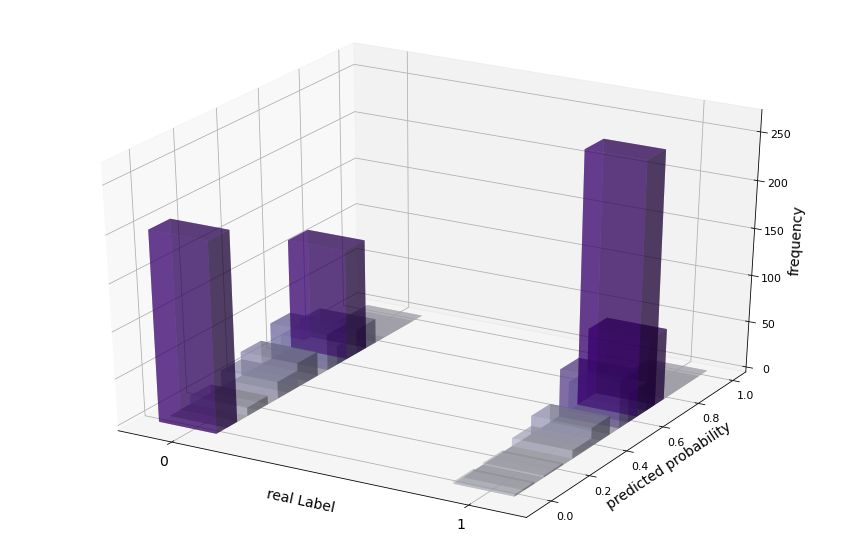}}}
\subfigure[Germany Test Set]{{\includegraphics[width=5.61cm]{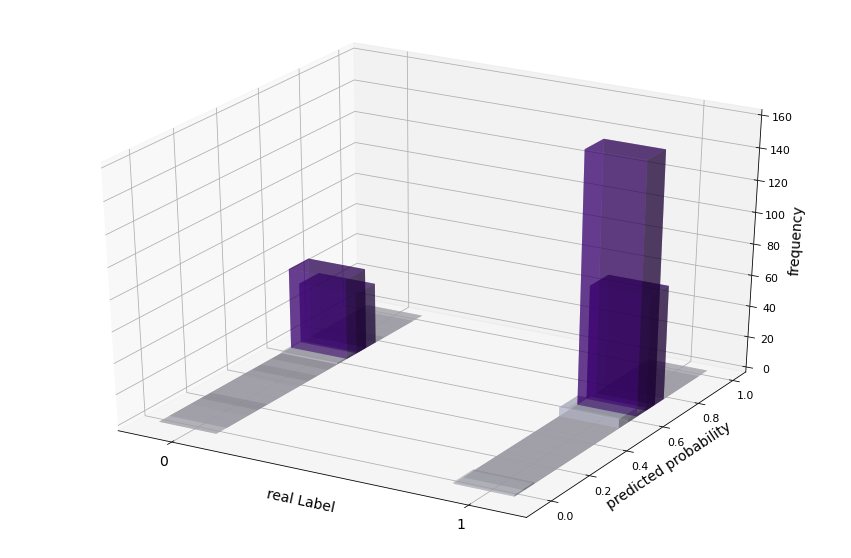} }}\\
\subfigure[Australia Train set]{{\includegraphics[width=5.61cm]{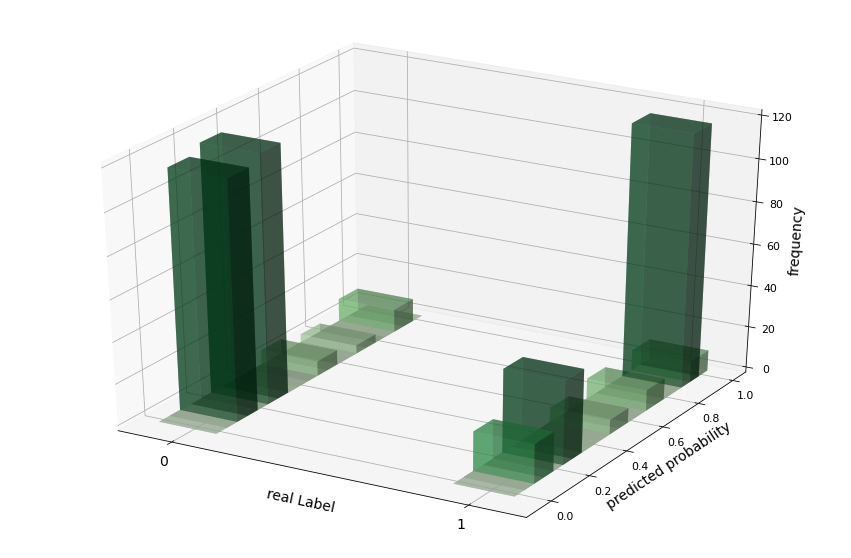}}}
\subfigure[Australia Test Set]{{\includegraphics[width=5.61cm]{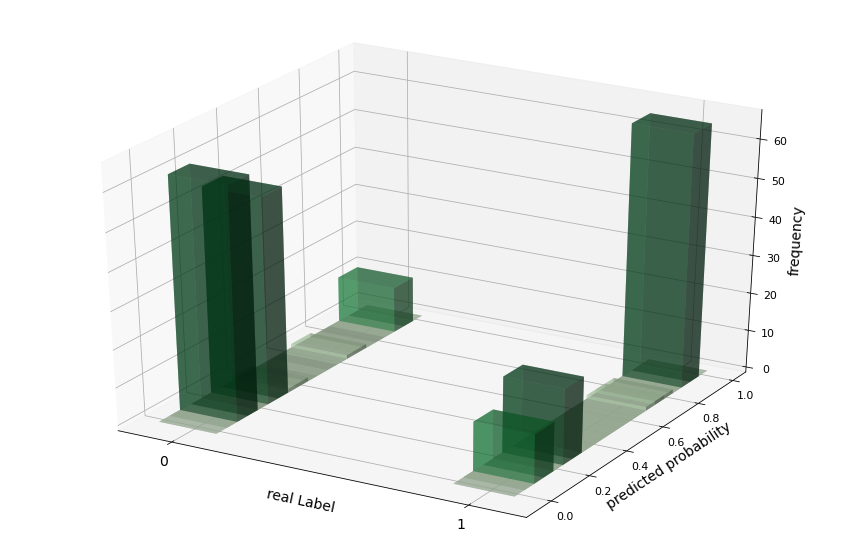} }}\\
\subfigure[Japan Train set]{{\includegraphics[width=5.61cm]{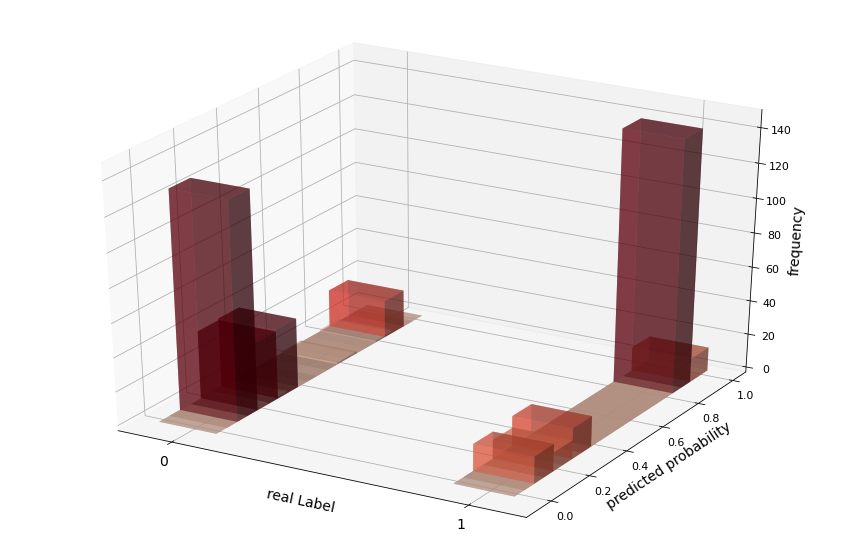}}}
\subfigure[Japan Test Set]{{\includegraphics[width=5.61cm]{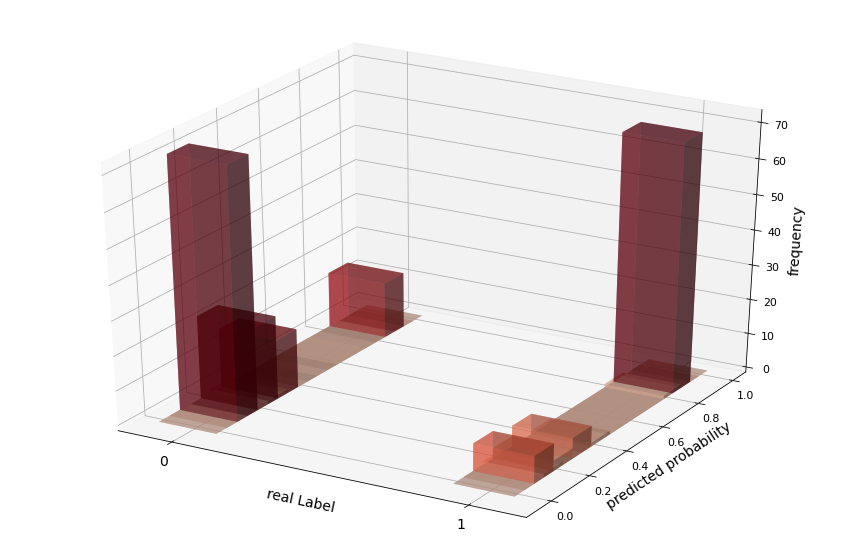} }}%
\caption{The probability of default for individuals German, Australian and Japanese data sets.}
\label{fig:probability_default_2}
\end{figure}

\begin{figure}
\centering
\subfigure[SBA (first 100 k selected group) Train set]{{\includegraphics[width=5.61cm]{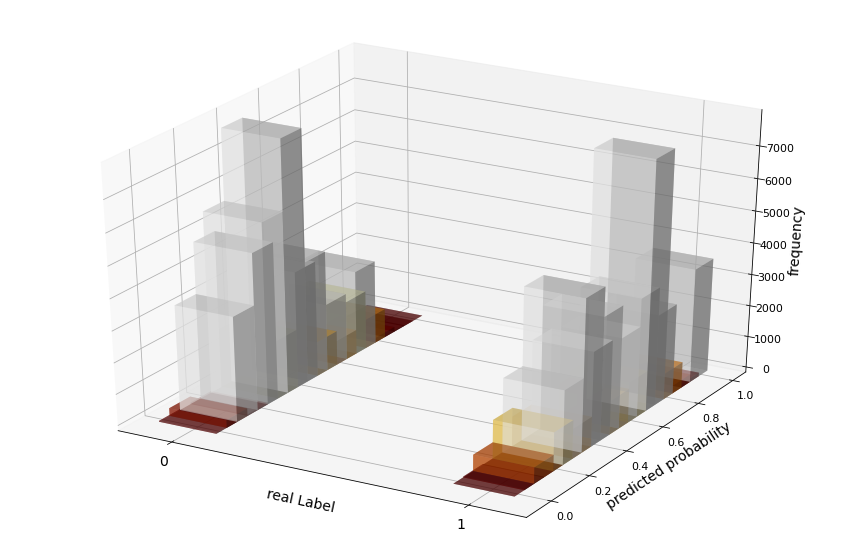}}}
\subfigure[SBA (first 100 k selected group) Test Set]{{\includegraphics[width=5.61cm]{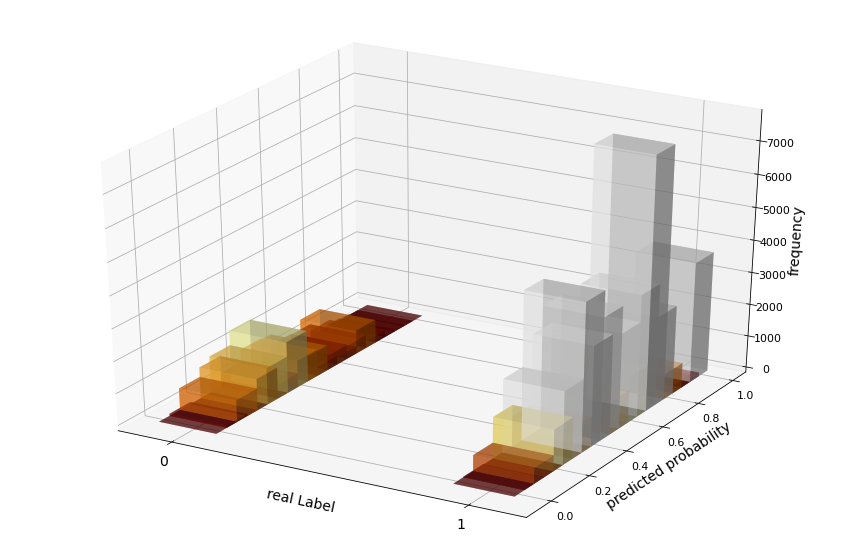} }}\\
\subfigure[SBA (second 100 k selected group) Train set]{{\includegraphics[width=5.61cm]{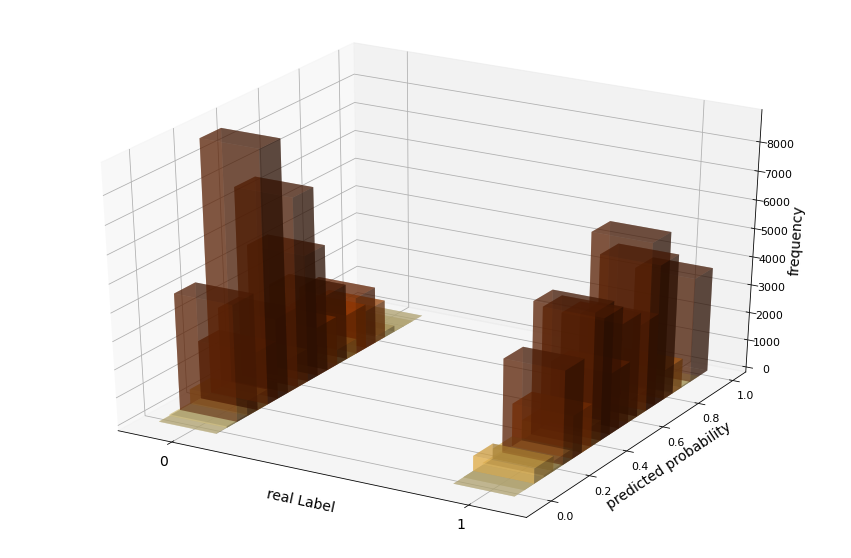}}}
\subfigure[SBA (second 100 k selected group) Test Set]{{\includegraphics[width=5.61cm]{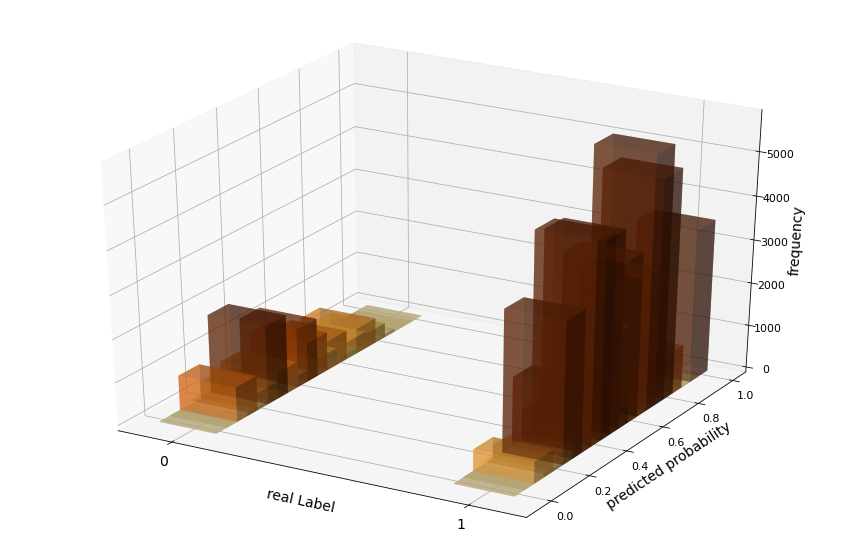} }}\\
\subfigure[SBA (third 100 k selected group) Train set]{{\includegraphics[width=5.61cm]{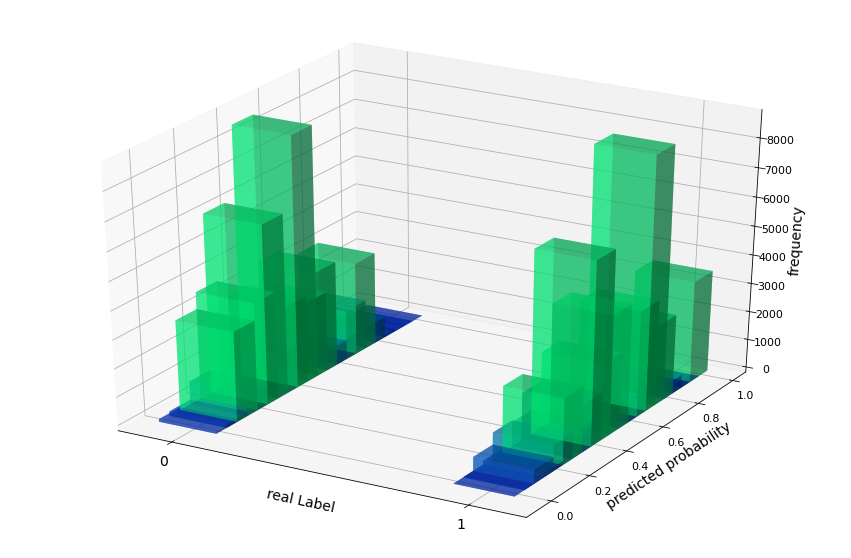}}}
\subfigure[SBA (third 100 k selected group) Test Set]{{\includegraphics[width=5.61cm]{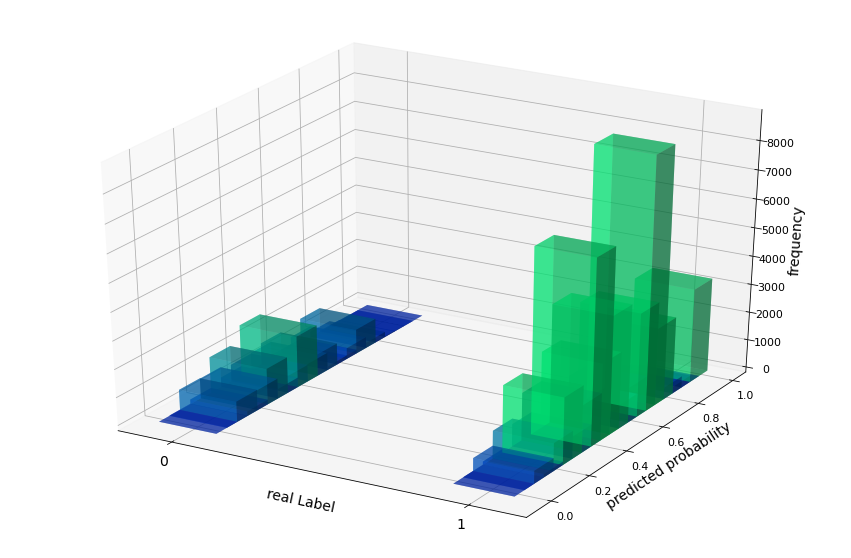} }}%
\caption{The probability of default for individuals German, Australian and Japanese data sets.}
\label{fig:probability_default_3}
\end{figure}


\begin{figure}
\centering
\subfigure[German train data]{{\includegraphics[width=5.40cm]{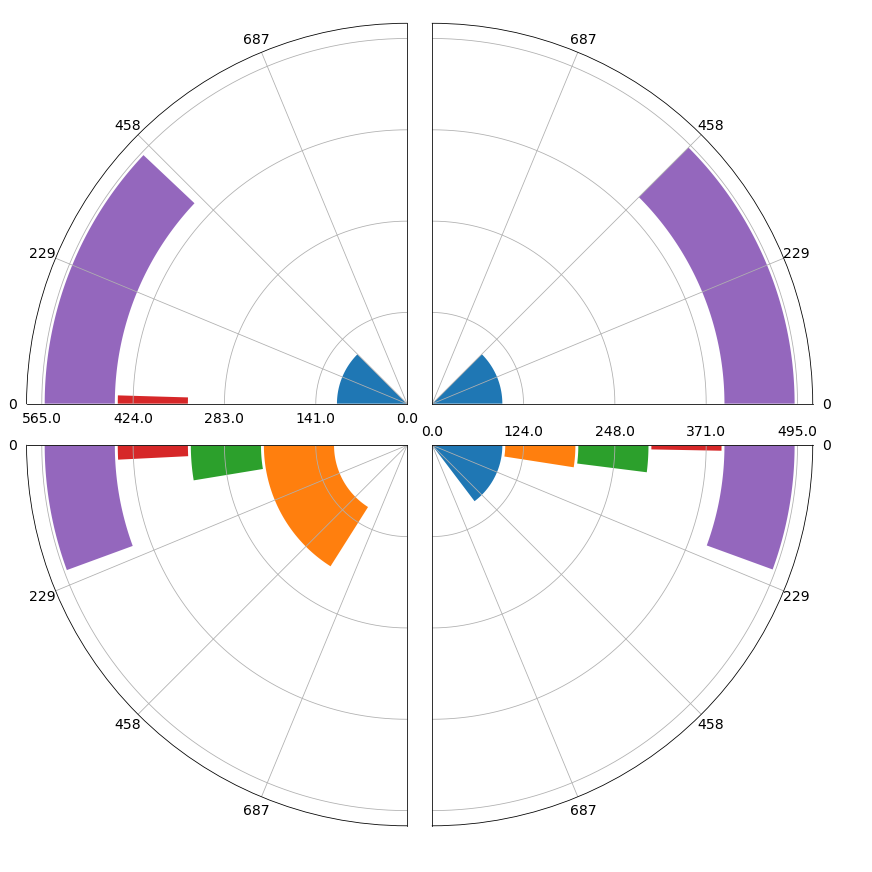}}}
\quad
\subfigure[German test data]{{\includegraphics[width=5.40cm]{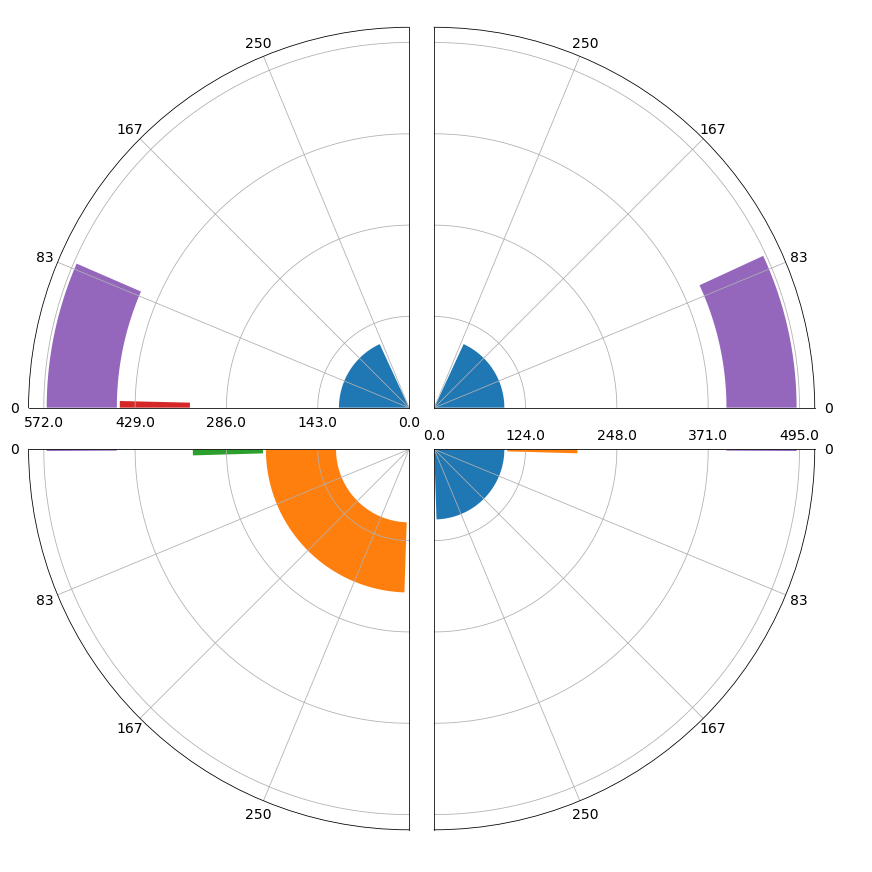} }}\\
\subfigure[Japanese train data]{{\includegraphics[width=5.40cm]{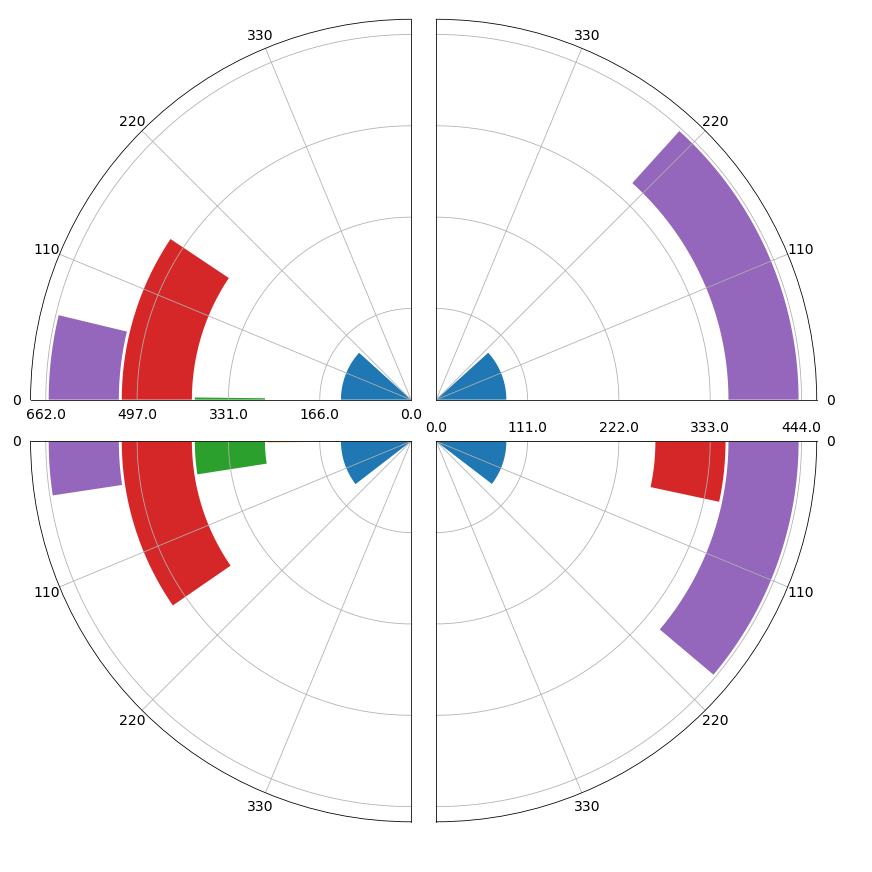}}}
\quad
\subfigure[Japanese test data]{{\includegraphics[width=5.40cm]{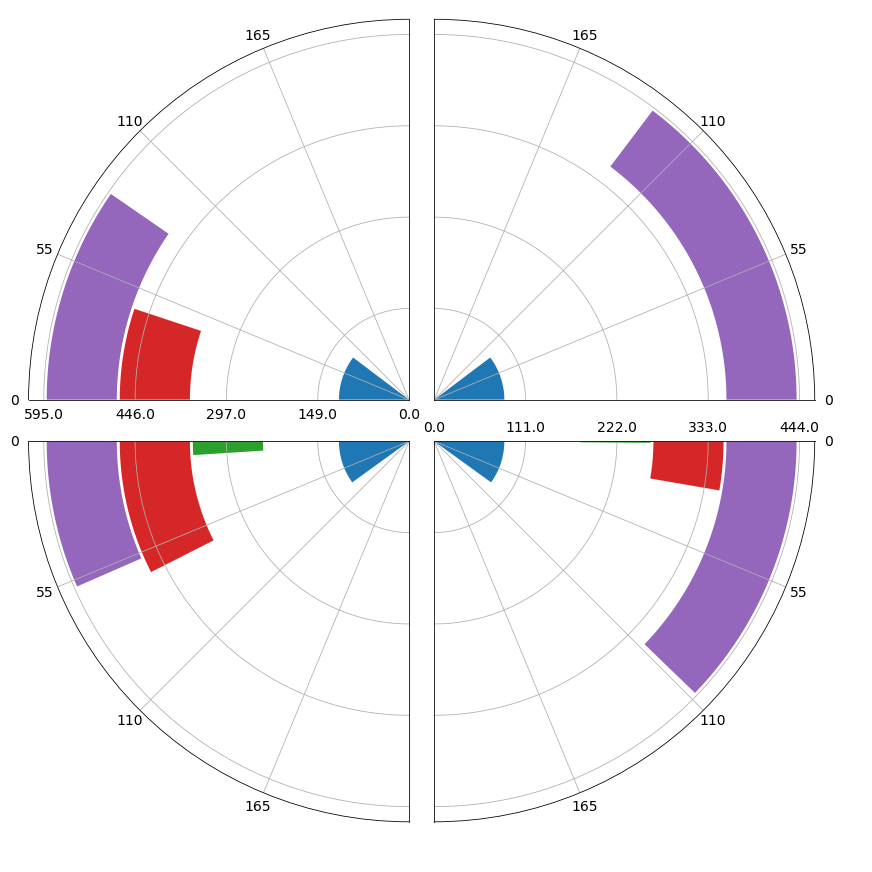} }}\\
\subfigure[Australian  train data]{{\includegraphics[width=5.40cm]{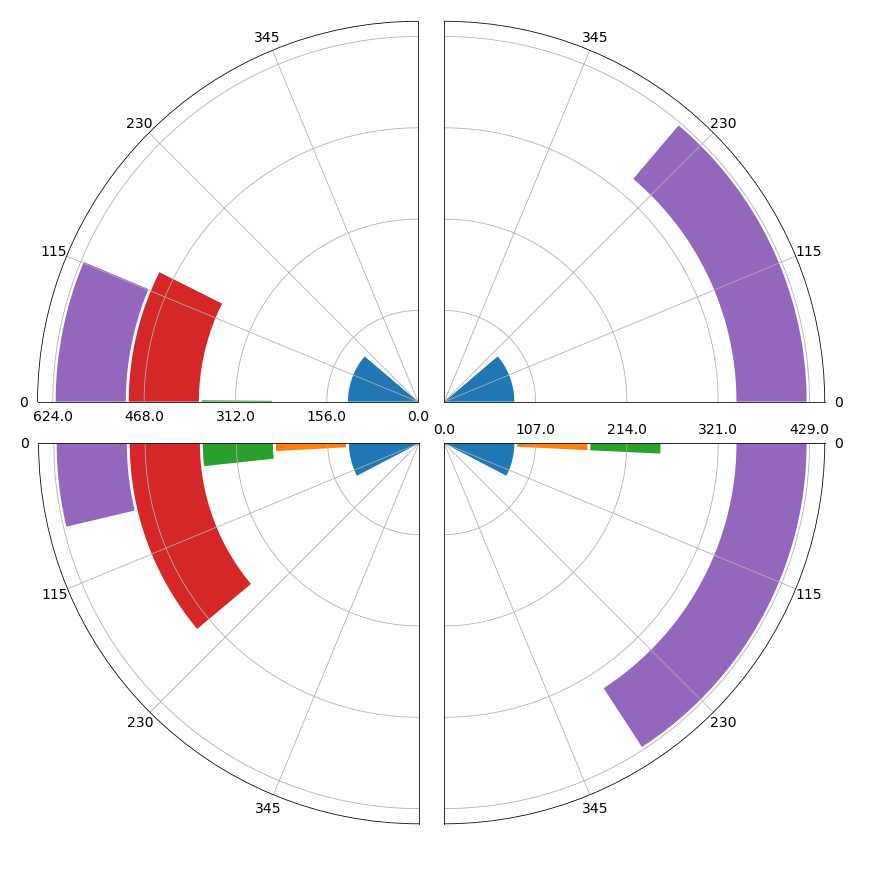}}}
\quad
\subfigure[Australian test data]{{\includegraphics[width=5.40cm]{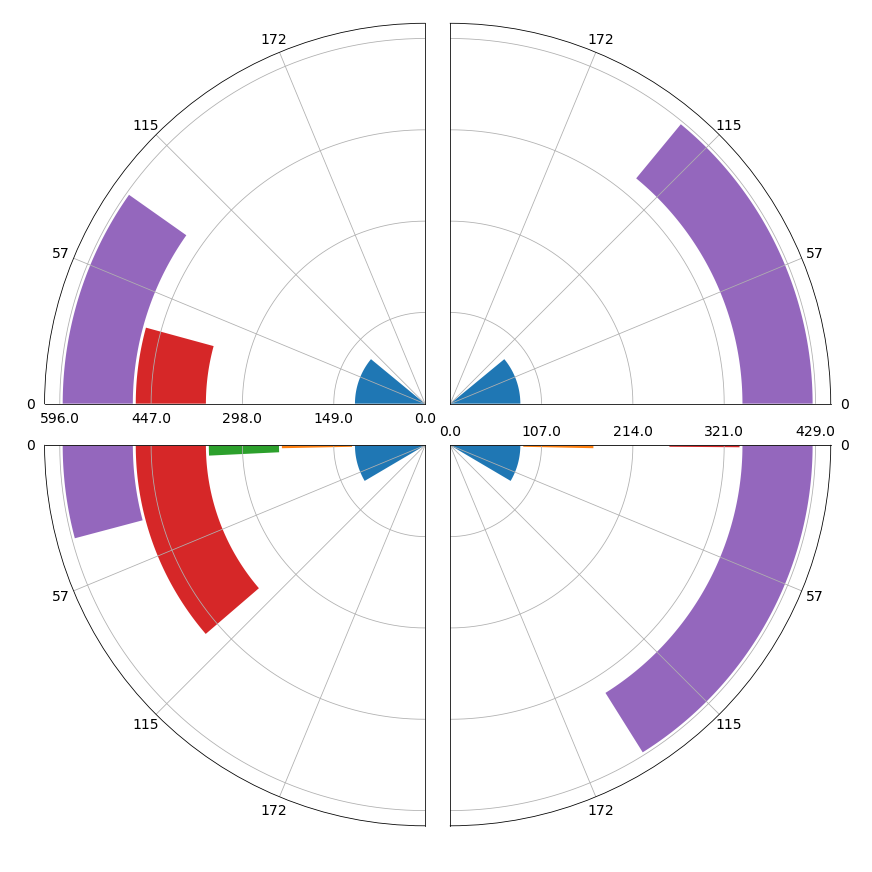} }}\\
\caption{The amount of actual loss versus predicted loss.}
\label{fig:fixed&random_loan}%
\end{figure}

 \begin{figure}
\centering
\subfigure[Telemarketing train data]{{\includegraphics[width=5.40cm]{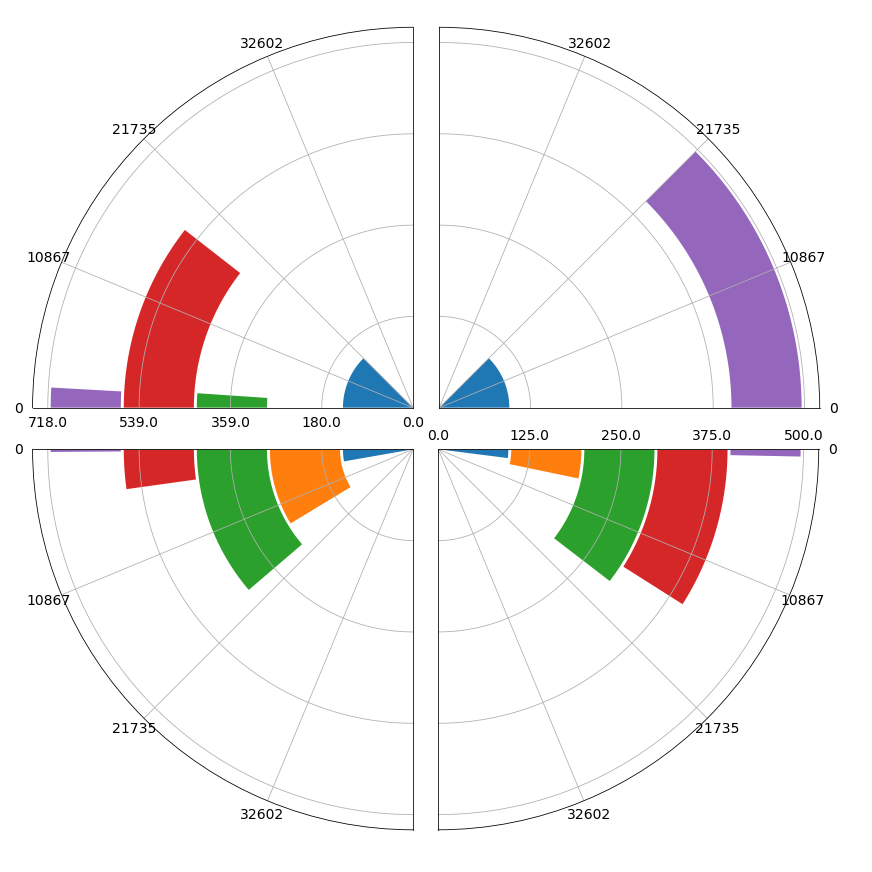}}}
\quad
\subfigure[Telemarketing test data]{{\includegraphics[width=5.40cm]{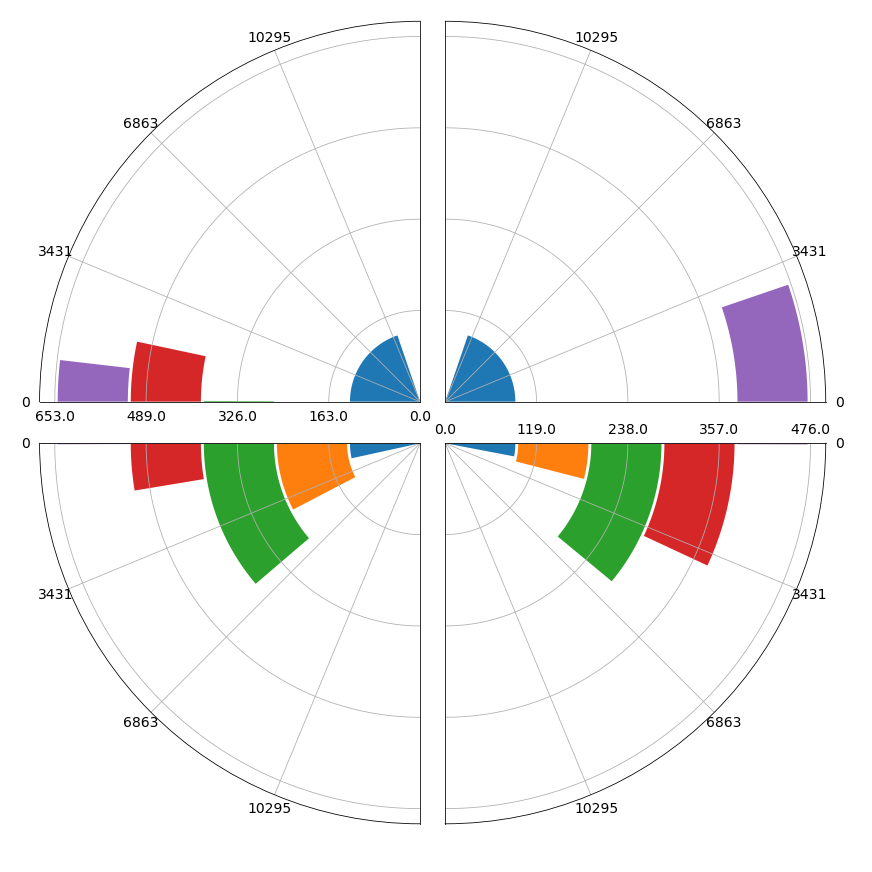} }}\\
\subfigure[SBA train data]{{\includegraphics[width=5.40cm]{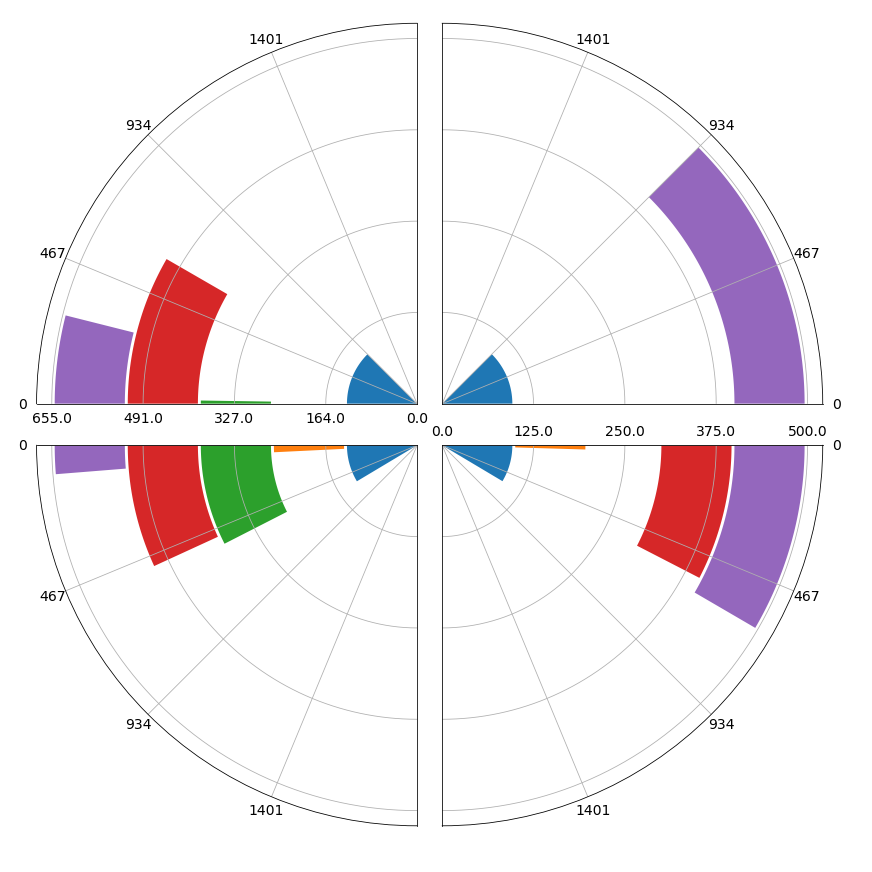}}}
\quad
\subfigure[SBA test data]{{\includegraphics[width=5.40cm]{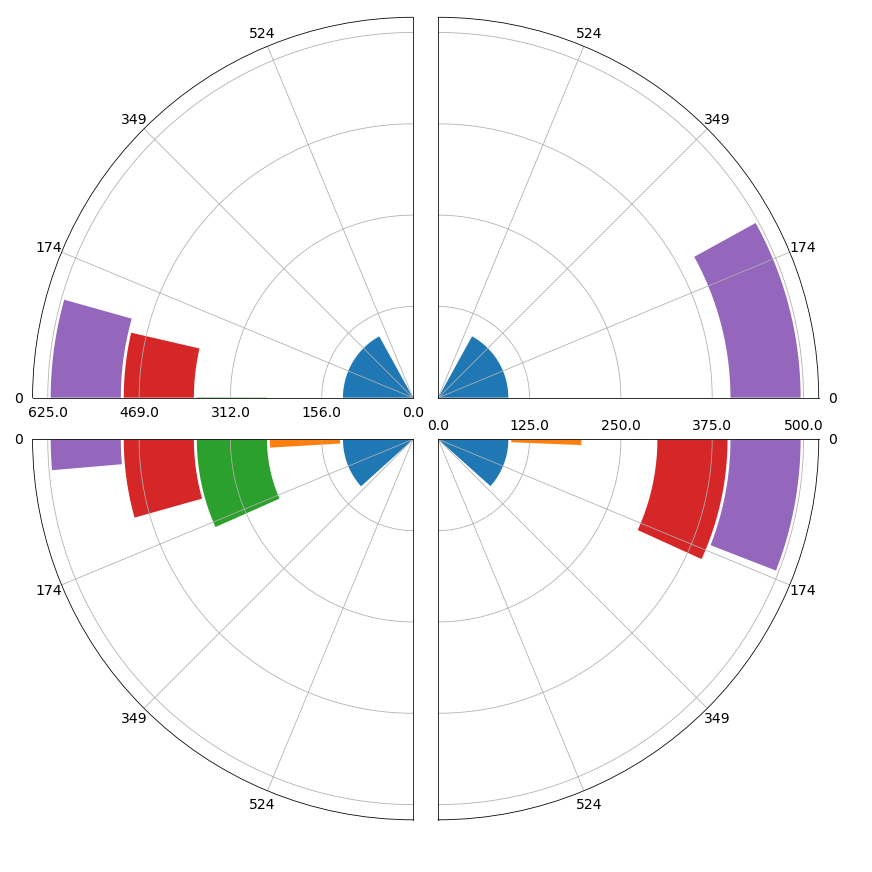} }}\\
\subfigure[Taiwanese  train data]{{\includegraphics[width=5.40cm]{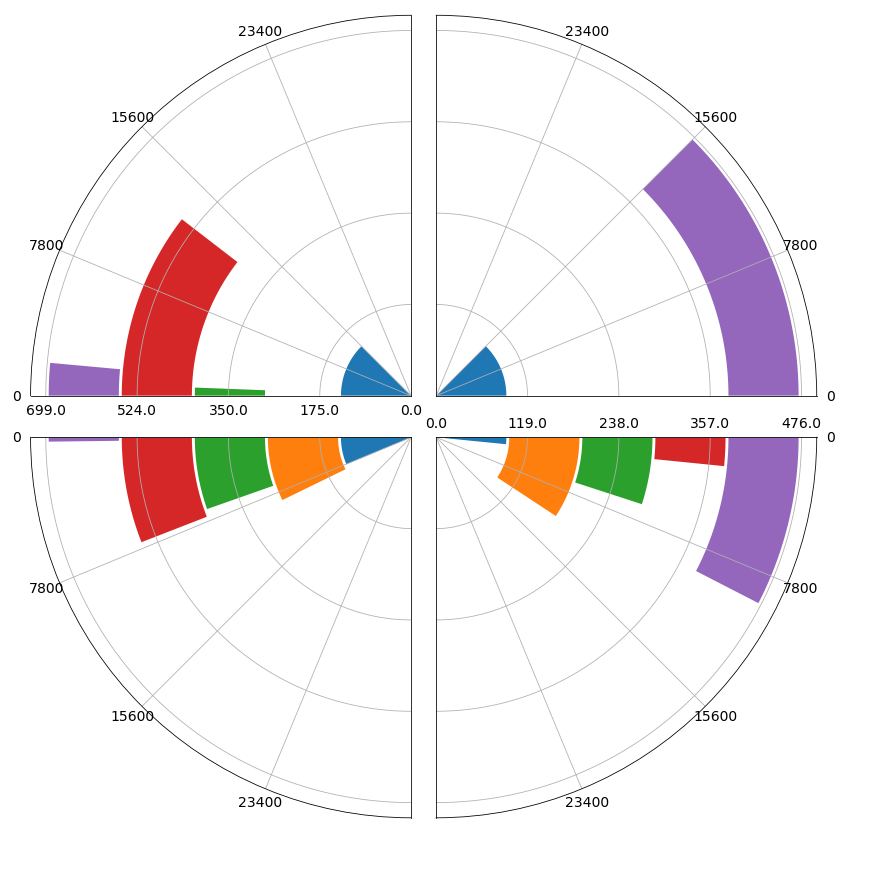}}}
\quad
\subfigure[Taiwanese test data]{{\includegraphics[width=5.40cm]{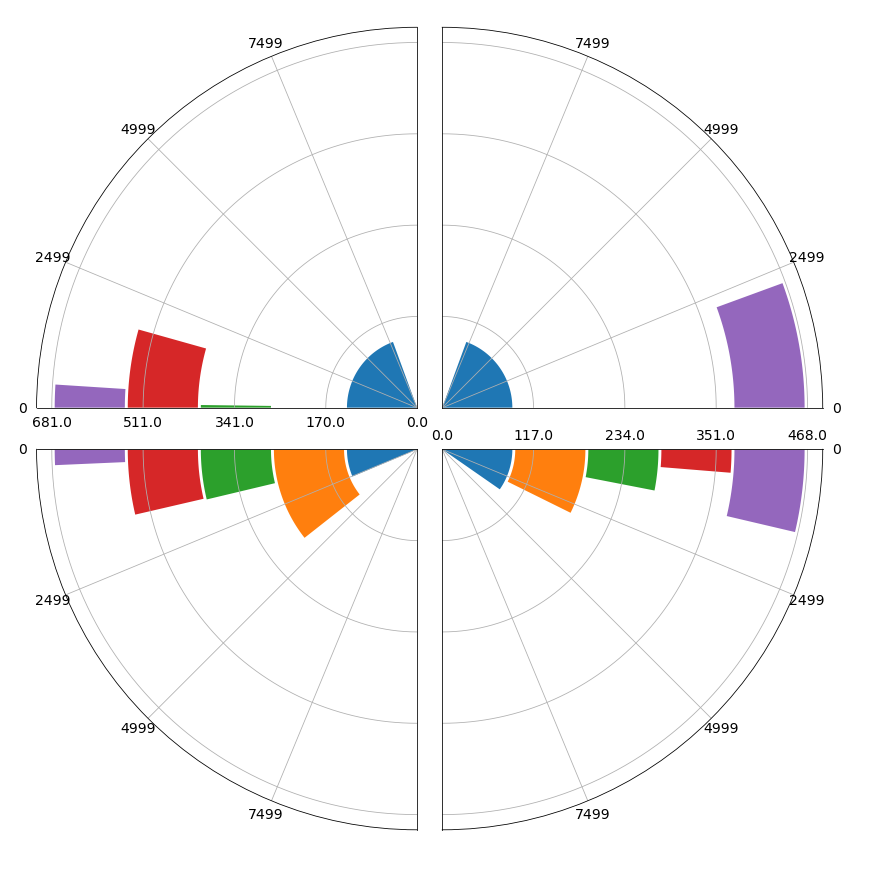} }}\\
\caption{The amount of actual loss versus predicted loss.}
\label{fig:fixed&random_loan}%
\end{figure}

 \begin{figure}
\centering
\subfigure[SBAk1 train data]{{\includegraphics[width=5.40cm]{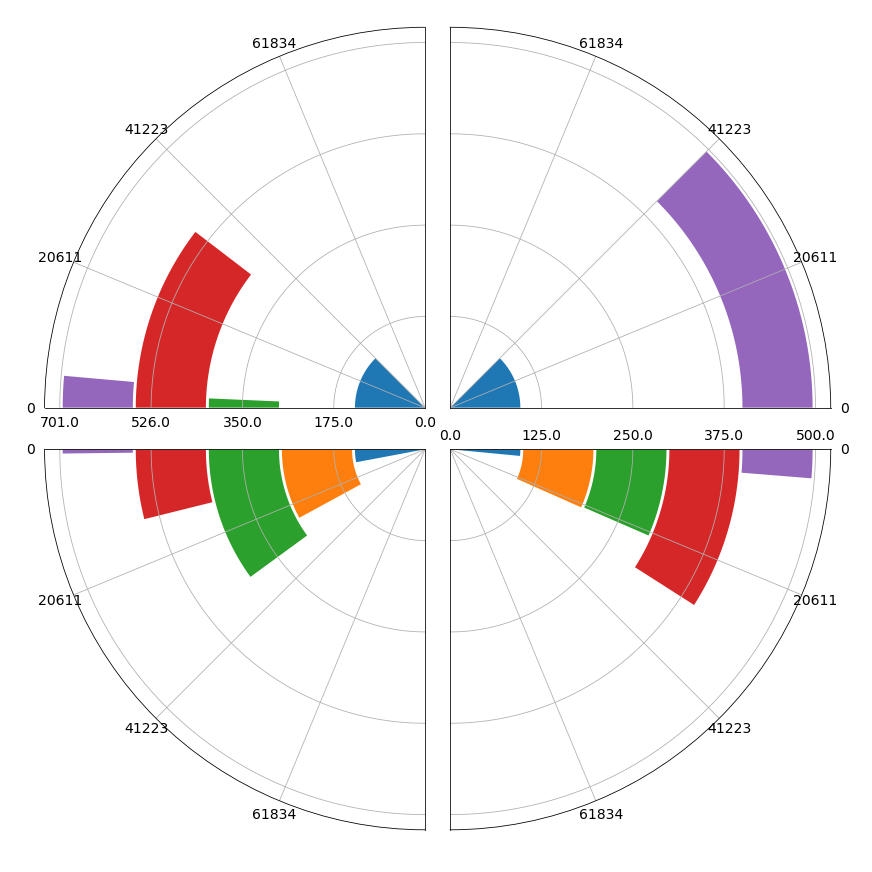}}}
\quad
\subfigure[SBAk1 test data]{{\includegraphics[width=5.40cm]{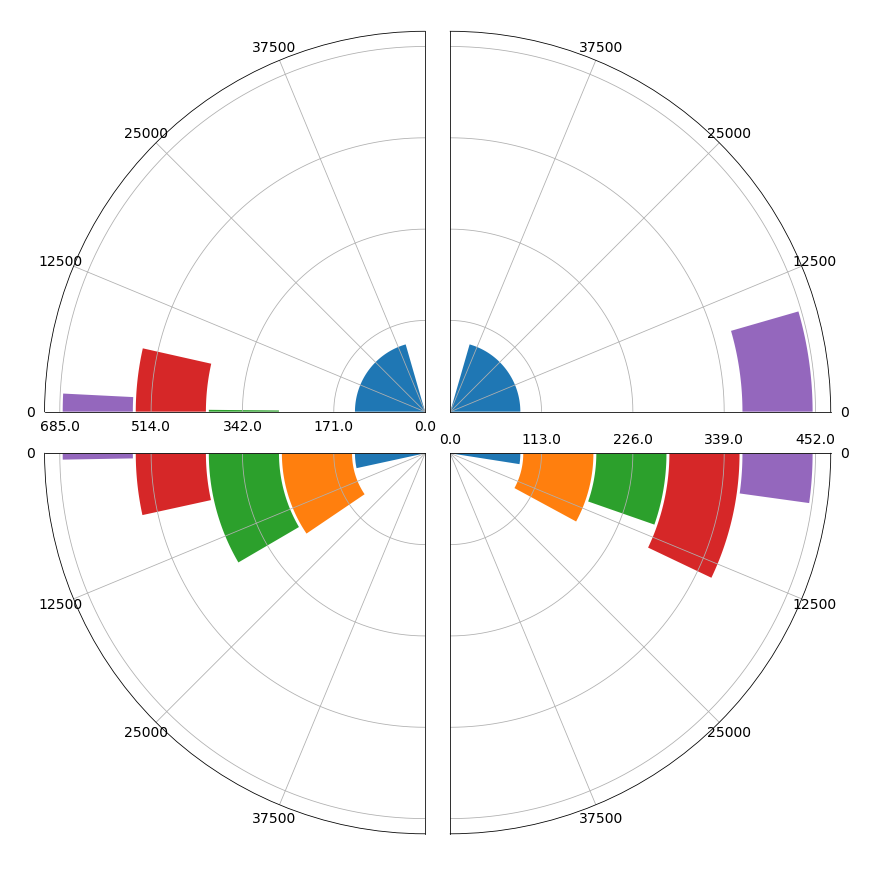} }}\\
\subfigure[SBAk2 train data]{{\includegraphics[width=5.40cm]{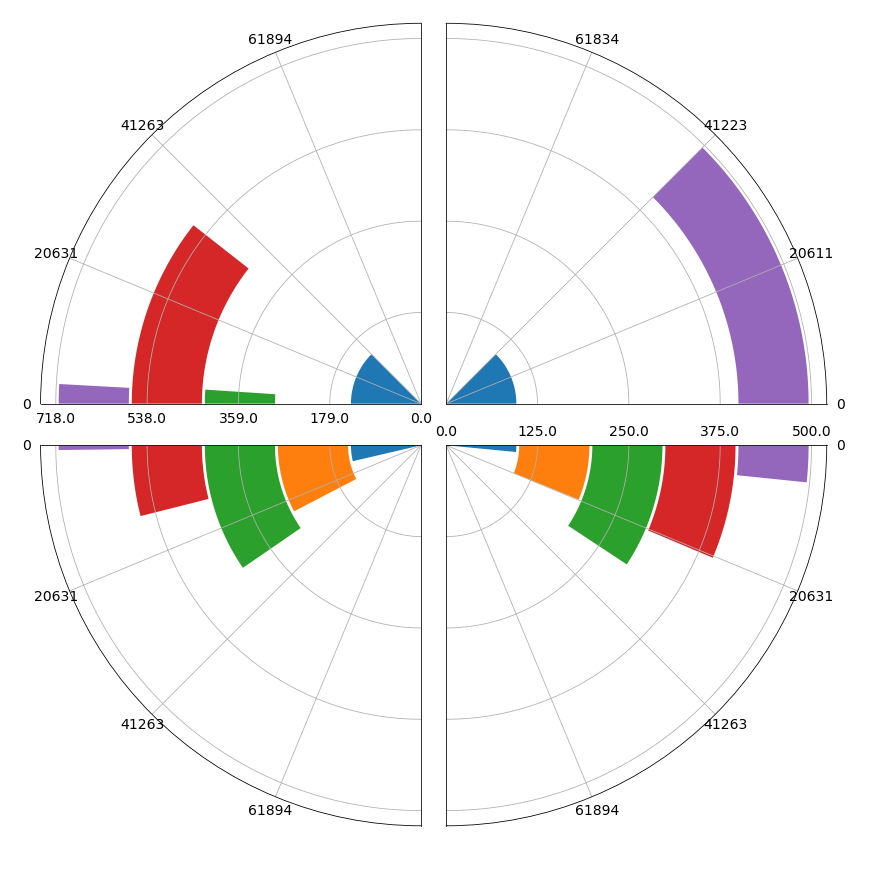}}}
\quad
\subfigure[SBAk2 test data]{{\includegraphics[width=5.40cm]{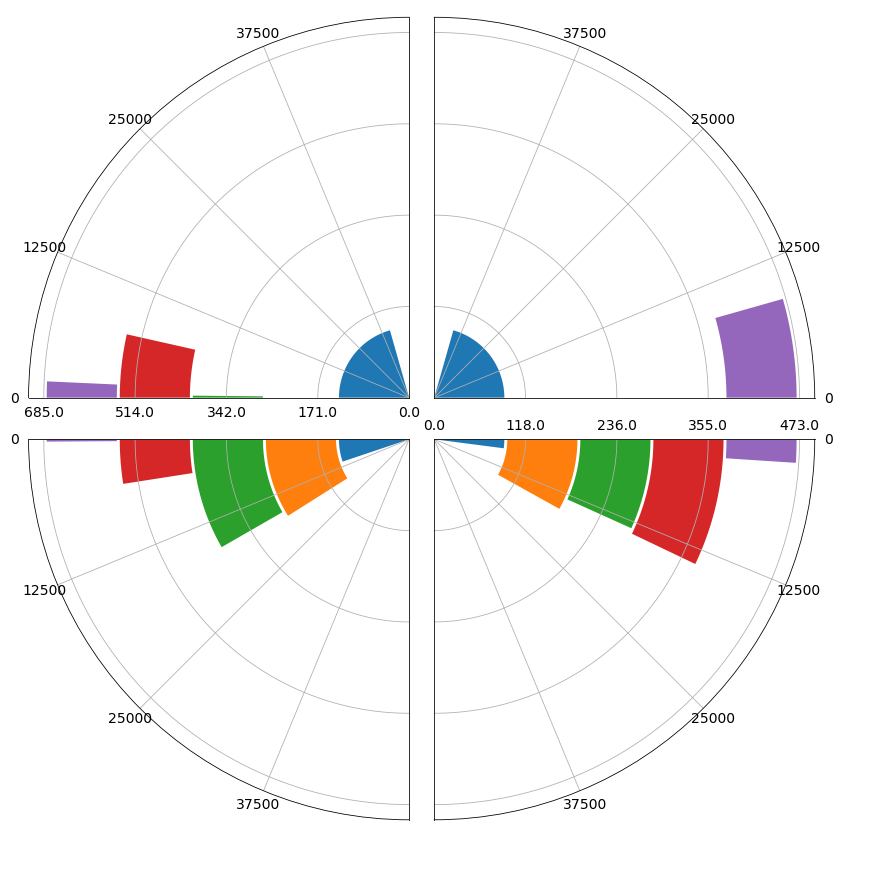} }}\\
\subfigure[SBAk2  train data]{{\includegraphics[width=5.40cm]{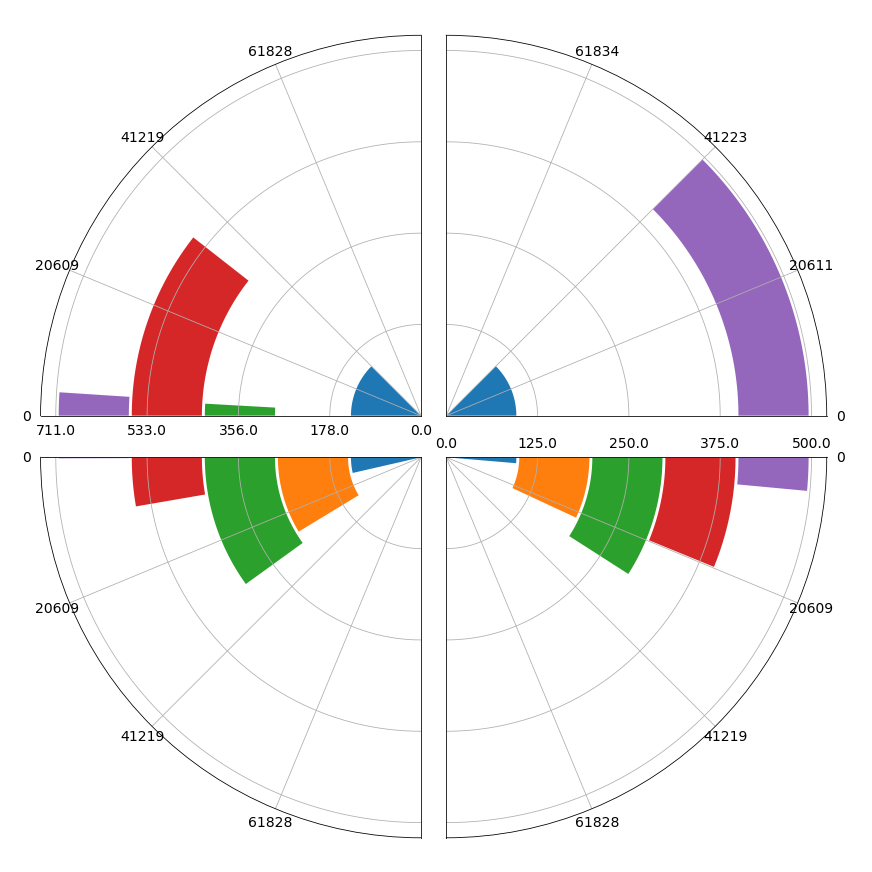}}}
\quad
\subfigure[SBAk3 test data]{{\includegraphics[width=5.40cm]{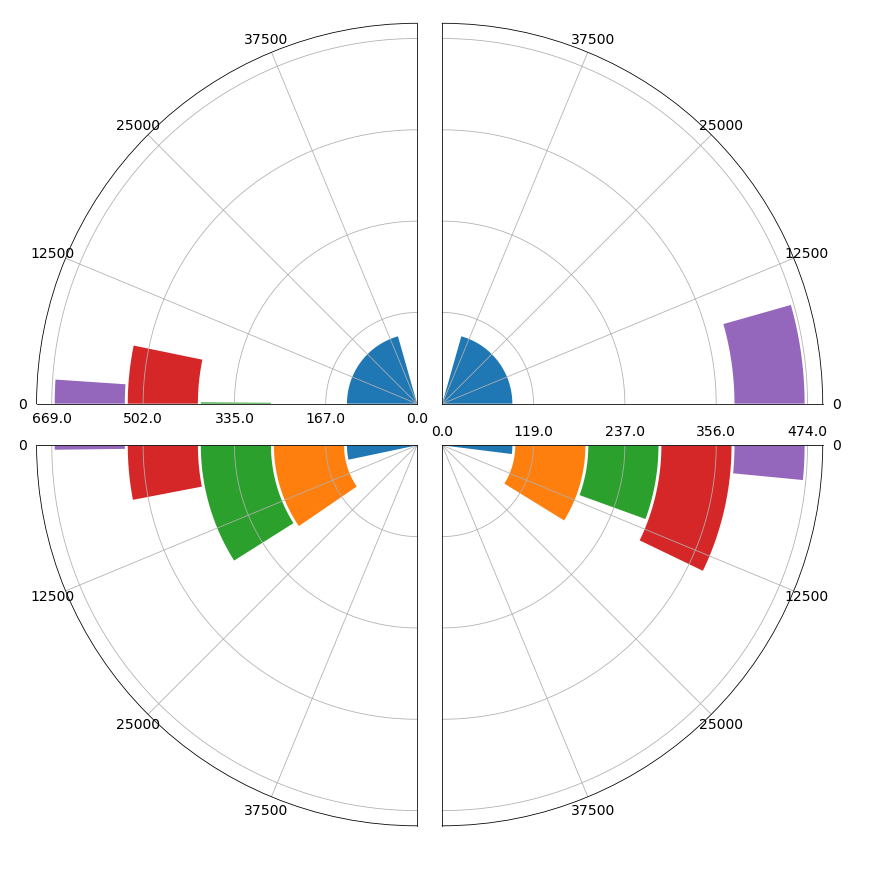} }}\\
\caption{The amount of actual loss versus predicted loss.}
\label{fig:fixed&random_loan}%
\end{figure}

\begin{table}
  \centering
\caption{Actual and predicted loss for different models and their errors}
\resizebox{\textwidth}{!}{
    \begin{tabular}{lrrrrrrrr}
    \toprule
    \boldmath{}\textbf{$M = 1000$ fixed (Train Set)}\unboldmath{} &       &       &       &       &       &       &       &  \\
          & \multicolumn{1}{c}{\textbf{Total}} & \multicolumn{1}{p{6.07em}}{\textbf{Actual}} & \multicolumn{1}{p{6.5em}}{\textbf{Predicted Loss}} & \multicolumn{1}{c}{\textbf{Relative}} & \multicolumn{1}{p{6.5em}}{\textbf{Predicted Loss}} & \multicolumn{1}{c}{\textbf{Relative}} & \multicolumn{1}{p{7.93em}}{\textbf{Predicted Loss}} & \multicolumn{1}{c}{\textbf{Relative}} \\
          & \multicolumn{1}{c}{\textbf{Capital}} & \multicolumn{1}{c}{\textbf{Loss}} & \multicolumn{1}{c}{\textbf{(GMM)}} & \multicolumn{1}{c}{\textbf{Error}} & \multicolumn{1}{c}{\textbf{(SVM)}} & \multicolumn{1}{c}{\textbf{Error}} & \multicolumn{1}{c}{\textbf{(LR)}} & \multicolumn{1}{c}{\textbf{Error}} \\
    \midrule
    \textbf{Germany} & 916,000 & 229,000 & 229,000 & 0.00\% & 228,252 & 0.08\% & 226,144 & 0.31\% \\
    \textbf{Japan} & 440,000 & 117,000 & 117,000 & 0.00\% & 115,716 & 0.29\% & 112,961 & 0.92\% \\
    \textbf{Australia} & 460,000 & 127,000 & 127,000 & 0.00\% & 125,412 & 0.35\% & 129,779 & 0.60\% \\
    \textbf{Taiwan} & 31,200,000 & 7,800,000 & 7,800,000 & 0.00\% & 7,817,018 & 0.05\% & 7,805,297 & 0.02\% \\
    \textbf{SBA} & 1,868,000 & 467,000 & 467,000 & 0.00\% & 458,198 & 0.47\% & 467,000 & 0.00\% \\
    \textbf{Telemarketing} & 43,470,000 & 10,867,500 & 10,867,500 & 0.00\% & 10,822,793 & 0.10\% & 10,869,983 & 0.01\% \\
    \textbf{SBA 50k1} & 82,446,000 & 20,611,500 & 20,611,500 & 0.00\% & 16,390,771 & 5.12\% & 20,618,452 & 0.01\% \\
    \textbf{SBA 50k2} & 82,526,000 & 20,631,500 & 20,631,500 & 0.00\% & 12,538,894 & 9.81\% & 20,637,919 & 0.01\% \\
    \textbf{SBA 50k3} & 82,438,000 & 20,609,500 & 20,609,500 & 0.00\% & 15,174,912 & 6.59\% & 20,615,872 & 0.01\% \\
          &       &       &       &       &       &       &       &  \\
          \midrule
    \boldmath{}\textbf{$M = 1000$ fixed (Train Set)}\unboldmath{} &       &       &       &       &       &       &       &  \\
          & \multicolumn{1}{c}{\textbf{Total}} & \multicolumn{1}{p{6.07em}}{\textbf{Actual}} & \multicolumn{1}{p{6.5em}}{\textbf{Predicted Loss}} & \multicolumn{1}{c}{\textbf{Relative}} & \multicolumn{1}{p{6.5em}}{\textbf{Predicted Loss}} & \multicolumn{1}{c}{\textbf{Relative}} & \multicolumn{1}{p{7.93em}}{\textbf{Predicted Loss}} & \multicolumn{1}{c}{\textbf{Relative}} \\
          & \multicolumn{1}{c}{\textbf{Capital}} & \multicolumn{1}{c}{\textbf{Loss}} & \multicolumn{1}{c}{\textbf{(GMM)}} & \multicolumn{1}{c}{\textbf{Error}} & \multicolumn{1}{c}{\textbf{(SVM)}} & \multicolumn{1}{c}{\textbf{Error}} & \multicolumn{1}{c}{\textbf{(LR)}} & \multicolumn{1}{c}{\textbf{Error}} \\
    \midrule
    \textbf{Germany} & 334,000 & 46,000 & 49,144 & 0.94\% & 83,527 & 11.24\% & 55,510 & 2.85\% \\
    \textbf{Japan} & 221,000 & 65,000 & 61,346 & 1.65\% & 57,749 & 3.28\% & 57,762 & 3.28\% \\
    \textbf{Australia} & 230,000 & 64,500 & 63,364 & 0.49\% & 62,967 & 0.67\% & 66,799 & 1.00\% \\
    \textbf{Taiwan} & 9,999,000 & 1,117,500 & 2,008,574 & 8.91\% & 2,345,955 & 12.29\% & 2,378,428 & 12.61\% \\
    \textbf{SBA} & 699,000 & 112,000 & 135,868 & 3.41\% & 122,963 & 1.57\% & 111,860 & 0.02\% \\
    \textbf{Telemarketing} & 13,727,000 & 1,437,500 & 3,103,969 & 12.14\% & 3,296,921 & 13.55\% & 3,122,712 & 12.28\% \\
    \textbf{SBA 50k1} & 50,000,000 & 4,466,500 & 10,868,920 & 12.80\% & 10,013,849 & 11.09\% & 10,332,796 & 11.73\% \\
    \textbf{SBA 50k2} & 50,000,000 & 4,400,000 & 10,847,613 & 12.90\% & 7,649,808 & 6.50\% & 10,528,968 & 12.26\% \\
    \textbf{SBA 50k3} & 50,000,000 & 4,397,000 & 11,006,260 & 13.22\% & 9,285,423 & 9.78\% & 10,434,015 & 12.07\% \\
          &       &       &       &       &       &       &       &  \\
          \midrule
    \boldmath{}\textbf{$M \sim \mathcal{N}(1000, 100)$  (Train Set)}\unboldmath{} &       &       &       &       &       &       &       &  \\
          & \multicolumn{1}{c}{\textbf{Total}} & \multicolumn{1}{p{6.07em}}{\textbf{Actual}} & \multicolumn{1}{p{6.5em}}{\textbf{Predicted Loss}} & \multicolumn{1}{c}{\textbf{Relative}} & \multicolumn{1}{p{6.5em}}{\textbf{Predicted Loss}} & \multicolumn{1}{c}{\textbf{Relative}} & \multicolumn{1}{p{7.93em}}{\textbf{Predicted Loss}} & \multicolumn{1}{c}{\textbf{Relative}} \\
          & \multicolumn{1}{c}{\textbf{Capital}} & \multicolumn{1}{c}{\textbf{Loss}} & \multicolumn{1}{c}{\textbf{(GMM)}} & \multicolumn{1}{c}{\textbf{Error}} & \multicolumn{1}{c}{\textbf{(SVM)}} & \multicolumn{1}{c}{\textbf{Error}} & \multicolumn{1}{c}{\textbf{(LR)}} & \multicolumn{1}{c}{\textbf{Error}} \\
          & \multicolumn{1}{c}{\textbf{(Approixmate)}} &       &       &       &       &       &       &  \\
    \midrule
    \textbf{Germany} & 916,000 & 228,345 & 228,366 & 0.00\% & 227,496 & 0.09\% & 225,407 & 0.32\% \\
    \textbf{Japan} & 440,000 & 117,453 & 117,382 & 0.02\% & 116,099 & 0.31\% & 113,370 & 0.93\% \\
    \textbf{Australia} & 460,000 & 127,285 & 127,314 & 0.01\% & 125,865 & 0.31\% & 130,140 & 0.62\% \\
    \textbf{Taiwan} & 31,200,000 & 7,786,223 & 7,786,399 & 0.00\% & 7,817,018 & 0.10\% & 7,805,297 & 0.06\% \\
    \textbf{SBA} & 1,868,000 & 467,176 & 466,700 & 0.03\% & 458,409 & 0.47\% & 467,322 & 0.01\% \\
    \textbf{Telemarketing} & 43,470,000 & 10,869,859 & 10,869,292 & 0.00\% & 10,822,516 & 0.11\% & 10,870,556 & 0.00\% \\
    \textbf{SBA 50k1} & 82,446,000 & 20,635,076 & 20,625,158 & 0.01\% & 16,395,034 & 5.14\% & 20,627,681 & 0.01\% \\
    \textbf{SBA 50k2} & 82,526,000 & 20,629,775 & 20,628,724 & 0.00\% & 12,539,906 & 9.80\% & 20,633,195 & 0.00\% \\
    \textbf{SBA 50k3} & 82,438,000 & 20,607,390 & 20,602,555 & 0.01\% & 14,967,292 & 6.84\% & 20,252,249 & 0.43\% \\
          &       &       &       &       &       &       &       &  \\
          \midrule
    \boldmath{}\textbf{$M \sim \mathcal{N}(1000, 100)$  (Test Set)}\unboldmath{} &       &       &       &       &       &       &       &  \\
          & \multicolumn{1}{c}{\textbf{Total}} & \multicolumn{1}{p{6.07em}}{\textbf{Actual}} & \multicolumn{1}{p{6.5em}}{\textbf{Predicted Loss}} & \multicolumn{1}{c}{\textbf{Relative}} & \multicolumn{1}{p{6.5em}}{\textbf{Predicted Loss}} & \multicolumn{1}{c}{\textbf{Relative}} & \multicolumn{1}{p{7.93em}}{\textbf{Predicted Loss}} & \multicolumn{1}{c}{\textbf{Relative}} \\
          & \multicolumn{1}{c}{\textbf{Capital}} & \multicolumn{1}{c}{\textbf{Loss}} & \multicolumn{1}{c}{\textbf{(GMM)}} & \multicolumn{1}{c}{\textbf{Error}} & \multicolumn{1}{c}{\textbf{(SVM)}} & \multicolumn{1}{c}{\textbf{Error}} & \multicolumn{1}{c}{\textbf{(LR)}} & \multicolumn{1}{c}{\textbf{Error}} \\
          & \multicolumn{1}{c}{\textbf{(Approixmate)}} &       &       &       &       &       &       &  \\
    \midrule
    \textbf{Germany} & 334,000 & 46,214 & 49,084 & 0.86\% & 83,452 & 11.15\% & 55,629 & 2.82\% \\
    \textbf{Japan} & 221,000 & 65,086 & 61,994 & 1.40\% & 58,370 & 3.04\% & 58,217 & 3.11\% \\
    \textbf{Australia} & 230,000 & 64,196 & 63,187 & 0.44\% & 62,669 & 0.66\% & 66,643 & 1.06\% \\
    \textbf{Taiwan} & 9,999,000 & 1,119,712 & 2,011,281 & 8.92\% & 2,348,869 & 12.29\% & 2,380,581 & 12.61\% \\
    \textbf{SBA} & 699,000 & 112,197 & 135,952 & 3.40\% & 123,251 & 1.58\% & 112,040 & 0.02\% \\
    \textbf{Telemarketing} & 13,727,000 & 59,779 & 154,197 & 0.69\% & 3,302,660 & 23.62\% & 3,127,802 & 22.35\% \\
    \textbf{SBA 50k1} & 50,000,000 & 4,464,741 & 10,869,494 & 12.81\% & 10,011,965 & 11.09\% & 10,330,776 & 11.73\% \\
    \textbf{SBA 50k2} & 50,000,000 & 4,403,849 & 10,851,358 & 12.90\% & 7,656,332 & 6.50\% & 10,537,883 & 12.27\% \\
    \textbf{SBA 50k3} & 50,000,000 & 4,392,787 & 10,997,210 & 13.21\% & 9,280,042 & 9.77\% & 10,424,647 & 12.06\% \\
    \bottomrule
    \end{tabular}
    }
  \label{tab_prob_loss}%
\end{table}%

\begin{figure}
\subfigure[Australia]{{\includegraphics[width=0.31\linewidth]{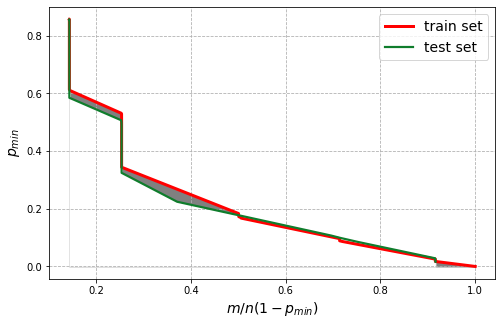}}}
\subfigure[Japan]{{\includegraphics[width=0.31\linewidth]{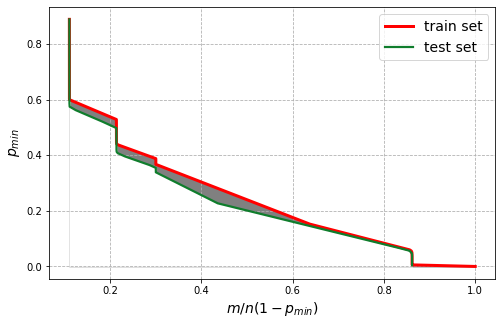} }}
\subfigure[Germany]{{\includegraphics[width=0.31\linewidth]{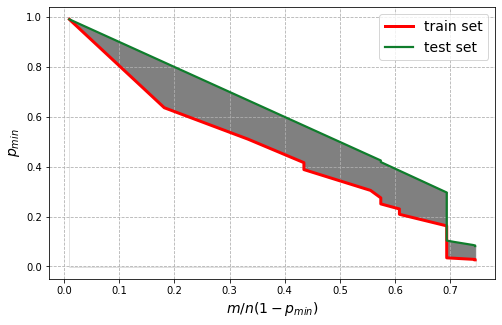}}} \\
\subfigure[Telemarketing]{{\includegraphics[width=0.31\linewidth]{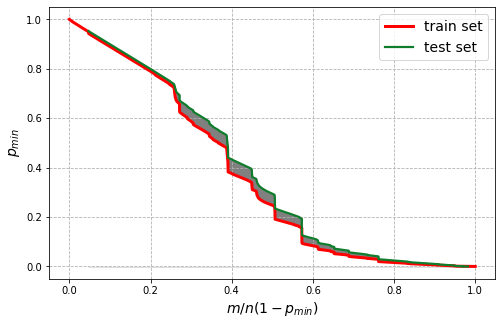}}}
\subfigure[Taiwan]{{\includegraphics[width=0.31\linewidth]{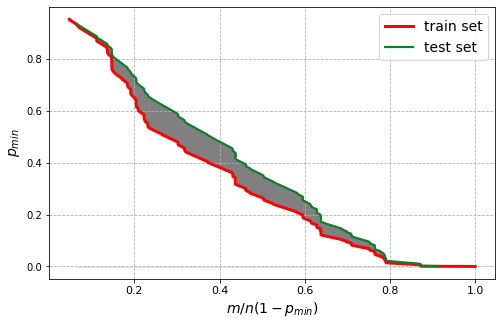} }}
\subfigure[SBA-case]{{\includegraphics[width=0.31\linewidth]{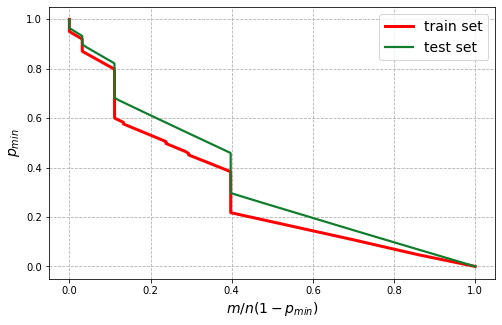}}} \\
\subfigure[SBA (first 100 k)]{{\includegraphics[width=0.31\linewidth]{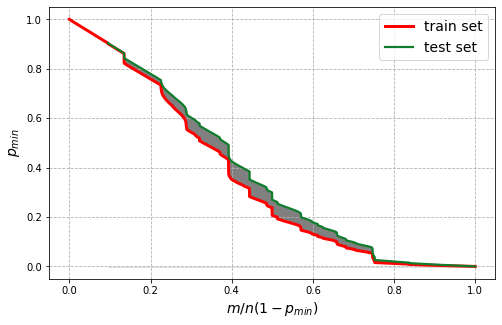}}}
\subfigure[SBA (second 100 k)]{{\includegraphics[width=0.31\linewidth]{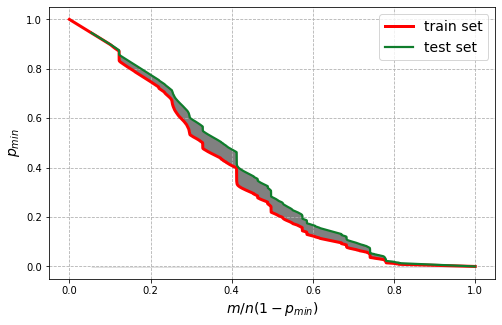} }}
\subfigure[SBA (third 100 k)]{{\includegraphics[width=0.31\linewidth]{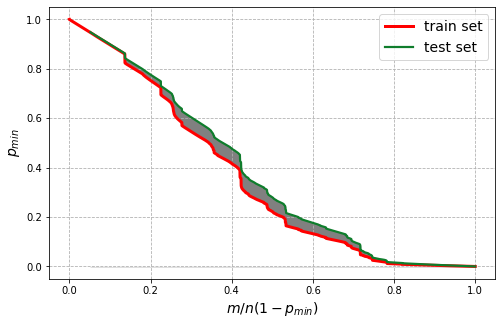}}} \\
\caption{The relationship between decision threshold and $\frac{m}{n}(1-p_{\text{min}})$ }
\label{fig:mn}%
\end{figure}

\begin{figure}
\subfigure[Australia]{{\includegraphics[width=0.31\linewidth]{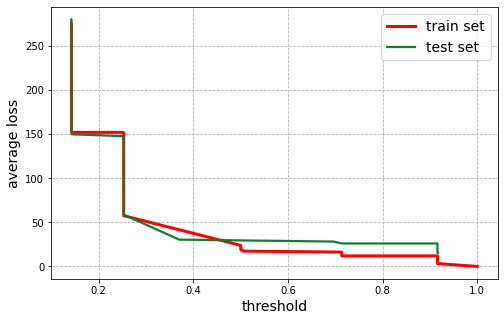}}}
\subfigure[Japan]{{\includegraphics[width=0.31\linewidth]{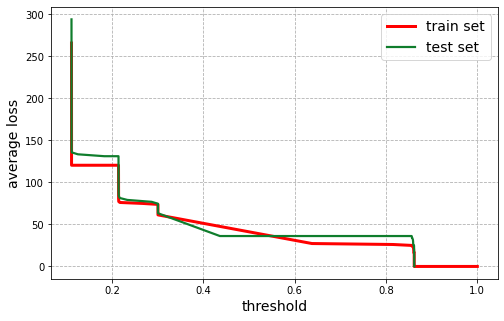} }}
\subfigure[Germany]{{\includegraphics[width=0.31\linewidth]{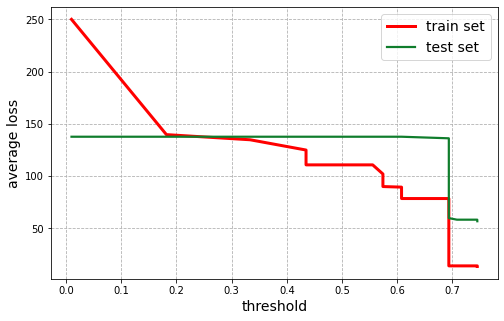}}} \\
\subfigure[Telemarketing]{{\includegraphics[width=0.31\linewidth]{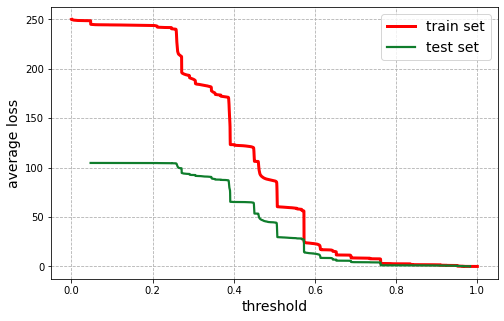}}}
\subfigure[Japan]{{\includegraphics[width=0.31\linewidth]{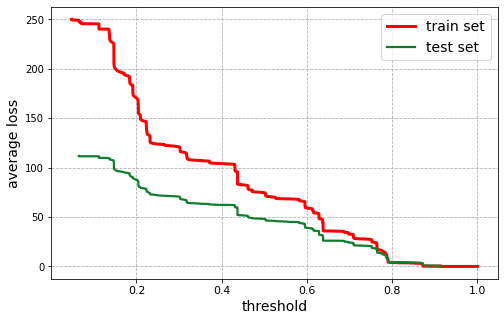} }}
\subfigure[Germany]{{\includegraphics[width=0.31\linewidth]{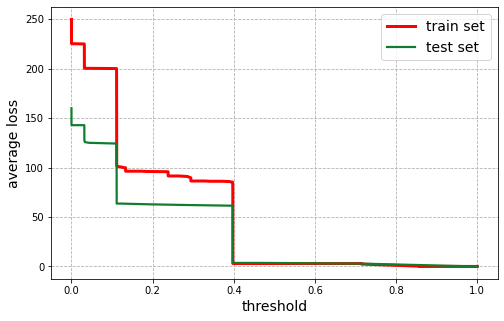}}}\\
\subfigure[SBA (first 100 k selected)]{{\includegraphics[width=0.31\linewidth]{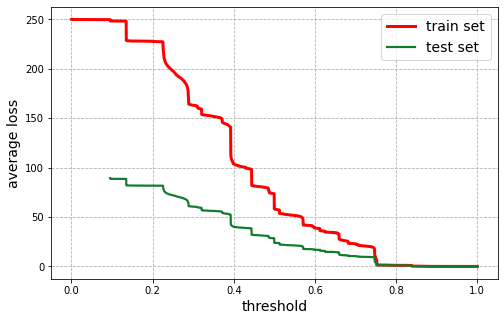}}}
\subfigure[SBA (second 100 k selected)]{{\includegraphics[width=0.31\linewidth]{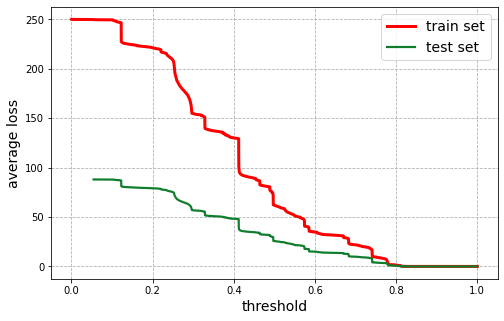} }}
\subfigure[SBA (third 100 k selected)]{{\includegraphics[width=0.31\linewidth]{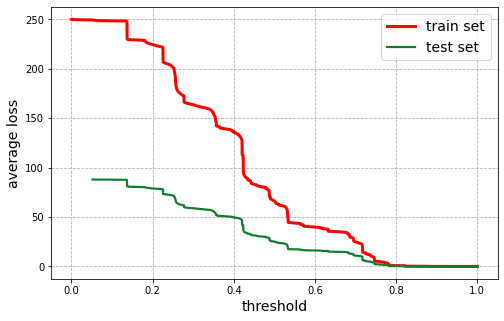}}}
\caption{The relationship between loss and $p_{\text{min}}$ } 

\label{fig:loss_fix}%
\end{figure}

\begin{figure}
\subfigure[Australia]{{\includegraphics[width=0.31\linewidth]{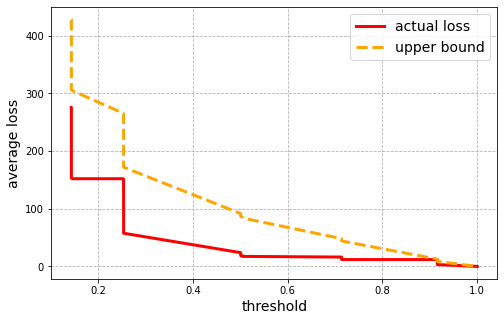}}}
\subfigure[Japan]{{\includegraphics[width=0.31\linewidth]{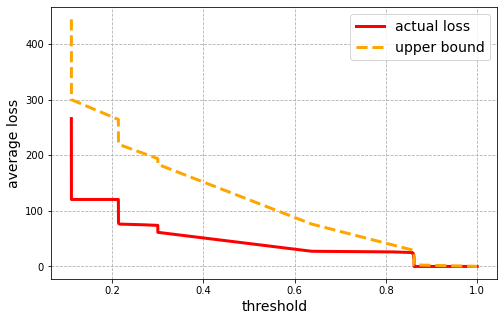} }}
\subfigure[Germany]{{\includegraphics[width=0.31\linewidth]{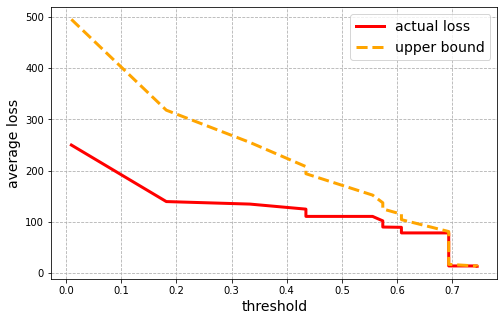}}}\\
\subfigure[Australia]{{\includegraphics[width=0.31\linewidth]{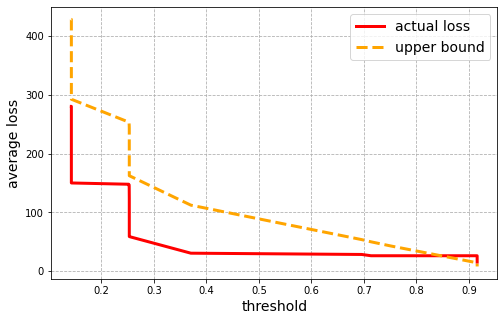}}}
\subfigure[Japan]{{\includegraphics[width=0.31\linewidth]{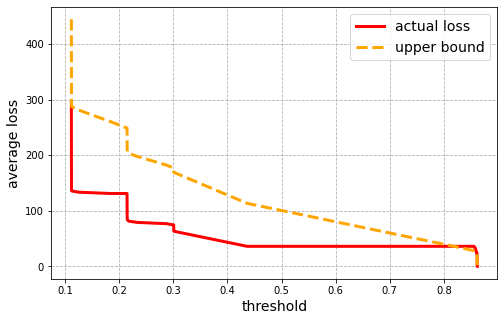} }}
\subfigure[Germany]{{\includegraphics[width=0.31\linewidth]{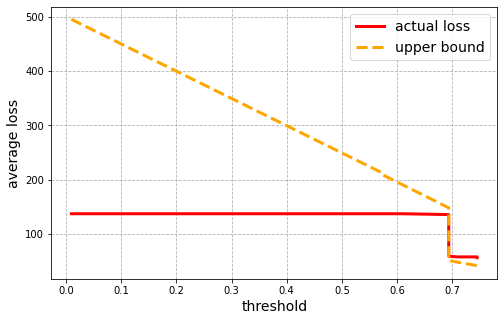}}}
\caption{The upper bound for the loss function using different decision thresholds (panels a, b and c represent the train set, while panels d, e and f show the test set results)}
\label{fig:upper_bound}%
\end{figure}

\begin{figure}
\subfigure[Telemarketing]{{\includegraphics[width=0.31\linewidth]{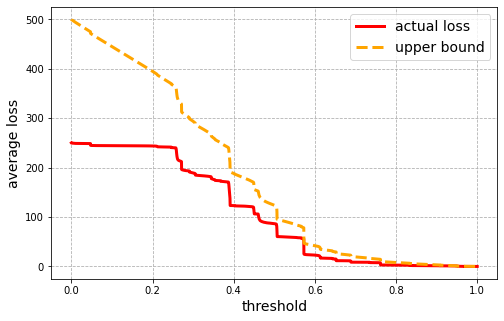}}}
\subfigure[Taiwan]{{\includegraphics[width=0.31\linewidth]{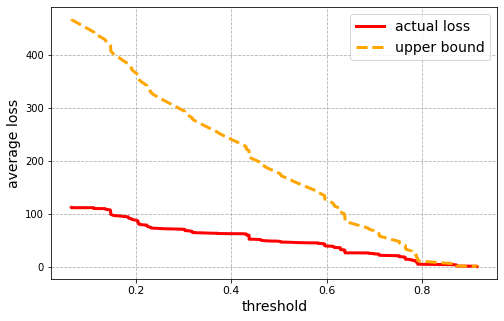} }}
\subfigure[SBA-case]{{\includegraphics[width=0.31\linewidth]{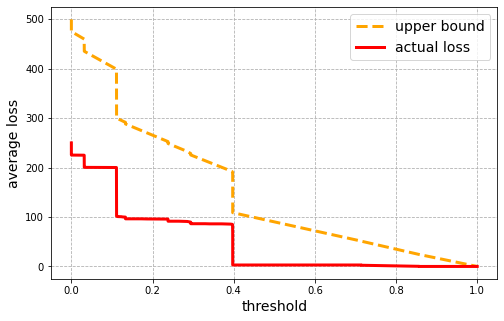}}}\\
\subfigure[Telemarketing]{{\includegraphics[width=0.31\linewidth]{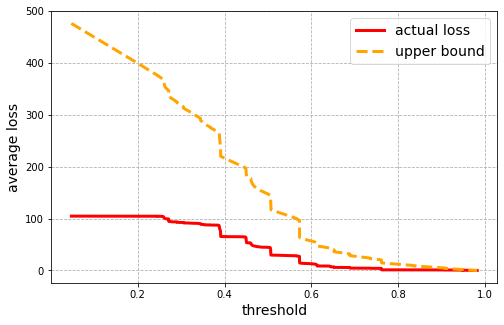}}}
\subfigure[Taiwan]{{\includegraphics[width=0.31\linewidth]{figs/chp3/taiwan_upper_test.png} }}
\subfigure[SBA-case]{{\includegraphics[width=0.31\linewidth]{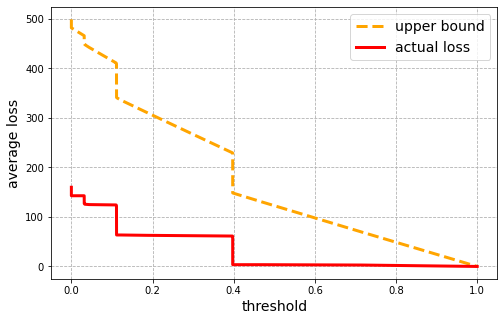}}}
\caption{The upper bound for the loss function using different decision thresholds (panels a, b and c represent the train set, while panels d, e and f show the test set results)}
\label{fig:upper_bound_2}%
\end{figure}

\begin{figure}
\subfigure[SBA 1]{{\includegraphics[width=0.31\linewidth]{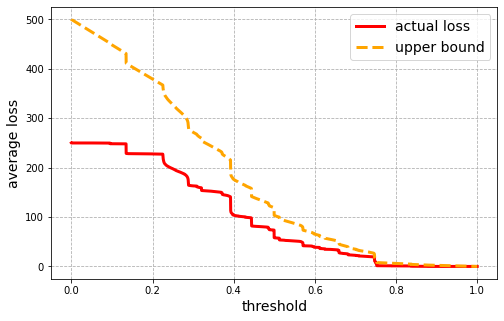}}}
\subfigure[SBA 2]{{\includegraphics[width=0.31\linewidth]{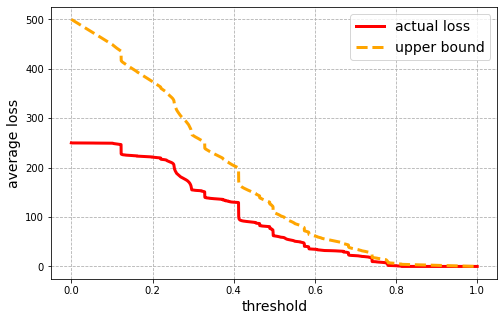} }}
\subfigure[SBA 3]{{\includegraphics[width=0.31\linewidth]{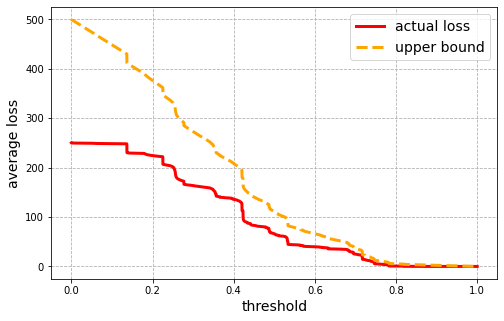}}}\\
\subfigure[SBA 1]{{\includegraphics[width=0.31\linewidth]{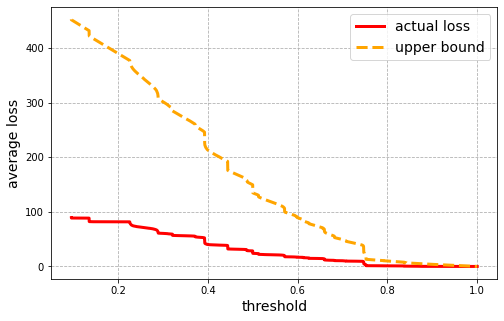}}}
\subfigure[SBA 2]{{\includegraphics[width=0.31\linewidth]{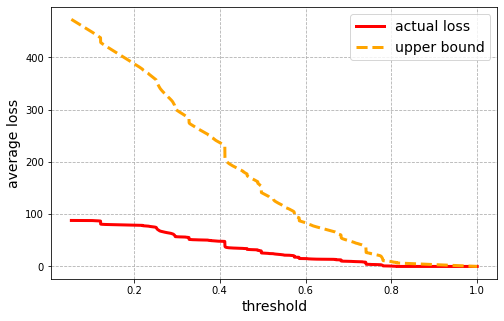} }}
\subfigure[SBA 3]{{\includegraphics[width=0.31\linewidth]{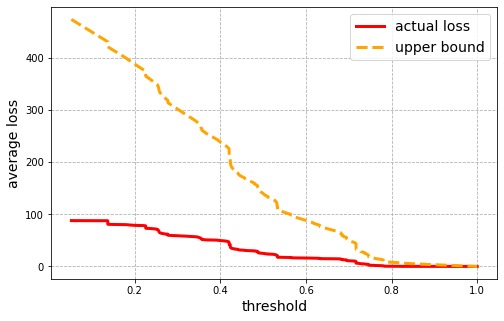}}}
\caption{The upper bound for the loss function using different decision thresholds (panels a, b and c represent the train set, while panels d, e and f show the test set results)}
\label{fig:upper_bound_3}%
\end{figure}

\begin{table}
\centering
\caption{Comparison of GMM's performance with SVM and LR through different accuracy measurements}
\resizebox{\textwidth}{!}{
\begin{tabular}{lrrrrrrrrrrrrrrr}
\toprule
\multicolumn{1}{c}{\multirow{2}[2]{*}{\textbf{Train}}} & \multicolumn{3}{c}{Accuracy} & \multicolumn{3}{c}{Precision score} & \multicolumn{3}{c}{recall score} & \multicolumn{3}{c}{F1 score} & \multicolumn{3}{c}{ROC AUC score} \\
      & \multicolumn{1}{c}{GMM} & \multicolumn{1}{c}{SVM} & \multicolumn{1}{c}{LR} & \multicolumn{1}{c}{GMM} & \multicolumn{1}{c}{SVM} & \multicolumn{1}{c}{LR} & \multicolumn{1}{c}{GMM} & \multicolumn{1}{c}{SVM} & \multicolumn{1}{c}{LR} & \multicolumn{1}{c}{GMM} & \multicolumn{1}{c}{SVM} & \multicolumn{1}{c}{LR} & \multicolumn{1}{c}{GMM} & \multicolumn{1}{c}{SVM} & \multicolumn{1}{c}{LR} \\
\midrule
Germany & 0.74  & 0.55  & 0.82  & 0.68  & 0.53  & 0.82  & 0.93  & 0.83  & 0.82  & 0.78  & 0.65  & 0.82  & 0.74  & 0.56  & 0.89 \\
Japan & 0.84  & 0.66  & 0.83  & 0.87  & 0.82  & 0.82  & 0.78  & 0.34  & 0.80  & 0.82  & 0.48  & 0.81  & 0.84  & 0.69  & 0.90 \\
Australia & 0.82  & 0.67  & 0.85  & 0.87  & 0.80  & 0.87  & 0.71  & 0.34  & 0.77  & 0.78  & 0.48  & 0.82  & 0.81  & 0.70  & 0.92 \\
Taiwan & 0.73  & 0.63  & 0.62  & 0.72  & 0.69  & 0.66  & 0.76  & 0.49  & 0.50  & 0.74  & 0.57  & 0.57  & 0.79  & 0.69  & 0.67 \\
SBA   & 0.85  & 0.88  & 0.99  & 0.98  & 0.97  & 0.99  & 0.71  & 0.78  & 0.99  & 0.82  & 0.86  & 0.99  & 0.94  & 0.98  & 1.00 \\
Telemarketing & 0.65  & 0.60  & 0.68  & 0.65  & 0.68  & 0.68  & 0.64  & 0.37  & 0.67  & 0.65  & 0.48  & 0.68  & 0.71  & 0.65  & 0.74 \\
SBA 50k1 & 0.68  & 0.44  & 0.78  & 0.72  & 0.46  & 0.82  & 0.59  & 0.78  & 0.73  & 0.65  & 0.58  & 0.77  & 0.75  & 0.44  & 0.82 \\
SBA 50k2 & 0.69  & 0.47  & 0.78  & 0.72  & 0.49  & 0.82  & 0.62  & 0.92  & 0.73  & 0.67  & 0.64  & 0.77  & 0.76  & 0.52  & 0.82 \\
SBA 50k3 & 0.68  & 0.46  & 0.78  & 0.70  & 0.47  & 0.82  & 0.62  & 0.84  & 0.72  & 0.66  & 0.61  & 0.77  & 0.75  & 0.41  & 0.82 \\
\midrule
\multicolumn{1}{c}{\multirow{2}[2]{*}{\textbf{Test}}} & \multicolumn{3}{c}{Accuracy} & \multicolumn{3}{c}{Precision score} & \multicolumn{3}{c}{recall score} & \multicolumn{3}{c}{F1 score} & \multicolumn{3}{c}{ROC AUC score} \\
      & \multicolumn{1}{c}{GMM} & \multicolumn{1}{c}{SVM} & \multicolumn{1}{c}{LR} & \multicolumn{1}{c}{GMM} & \multicolumn{1}{c}{SVM} & \multicolumn{1}{c}{LR} & \multicolumn{1}{c}{GMM} & \multicolumn{1}{c}{SVM} & \multicolumn{1}{c}{LR} & \multicolumn{1}{c}{GMM} & \multicolumn{1}{c}{SVM} & \multicolumn{1}{c}{LR} & \multicolumn{1}{c}{GMM} & \multicolumn{1}{c}{SVM} & \multicolumn{1}{c}{LR} \\
\midrule
Germany & 0.72  & 0.67  & 0.74  & 0.72  & 0.76  & 0.82  & 1.00  & 0.80  & 0.83  & 0.84  & 0.78  & 0.82  & 0.50  & 0.59  & 0.77 \\
Japan & 0.84  & 0.69  & 0.84  & 0.82  & 0.74  & 0.80  & 0.79  & 0.37  & 0.81  & 0.80  & 0.50  & 0.81  & 0.83  & 0.73  & 0.87 \\
Australia & 0.80  & 0.66  & 0.82  & 0.84  & 0.77  & 0.85  & 0.67  & 0.33  & 0.71  & 0.75  & 0.46  & 0.77  & 0.79  & 0.70  & 0.88 \\
Taiwan & 0.74  & 0.55  & 0.56  & 0.86  & 0.88  & 0.87  & 0.79  & 0.49  & 0.50  & 0.82  & 0.63  & 0.64  & 0.73  & 0.68  & 0.67 \\
SBA   & 0.80  & 0.85  & 0.99  & 0.99  & 0.99  & 0.99  & 0.71  & 0.79  & 1.00  & 0.83  & 0.88  & 0.99  & 0.94  & 0.99  & 1.00 \\
Telemarketing & 0.62  & 0.47  & 0.67  & 0.85  & 0.90  & 0.89  & 0.63  & 0.37  & 0.66  & 0.73  & 0.52  & 0.76  & 0.65  & 0.67  & 0.73 \\
SBA 50k1 & 0.62  & 0.66  & 0.75  & 0.91  & 0.80  & 0.95  & 0.60  & 0.78  & 0.73  & 0.72  & 0.79  & 0.83  & 0.73  & 0.44  & 0.81 \\
SBA 50k2 & 0.63  & 0.76  & 0.75  & 0.91  & 0.82  & 0.95  & 0.62  & 0.92  & 0.73  & 0.74  & 0.86  & 0.82  & 0.74  & 0.52  & 0.82 \\
SBA 50k3 & 0.64  & 0.70  & 0.74  & 0.91  & 0.81  & 0.95  & 0.62  & 0.84  & 0.72  & 0.74  & 0.82  & 0.82  & 0.73  & 0.42  & 0.82 \\
\bottomrule
\end{tabular}
}
\label{tab_accuracy2}
\end{table}

\end{document}